\documentclass[reprint,superscriptaddress,amsmath,amssymb,aps,pra,floatfix]{revtex4-2}

\usepackage{amsfonts}
\usepackage{amsmath}
\usepackage{dsfont}
\usepackage{amsthm}
\usepackage{mathrsfs}
\usepackage{mathtools}
\usepackage{amscd}
\usepackage{amssymb}
\usepackage{physics}

\usepackage[caption=false]{subfig}

\usepackage{amsxtra}
\usepackage{dcolumn}
\usepackage{bm}         
\usepackage{bbm}
\usepackage{graphicx}
\graphicspath{{./figures/}}
\usepackage{epstopdf}
\usepackage{color}
\usepackage[dvipsnames]{xcolor}
\usepackage[colorlinks,linkcolor=Blue,citecolor=Green]{hyperref}
\usepackage{textcomp}

\usepackage{booktabs} 
\usepackage{diagbox}   

\usepackage{microtype}

\newcommand{\la}{\langle}
\newcommand{\ra}{\rangle}

\newcommand{\da}{\dagger}

\newcommand{\Op}[1]{\hat{#1}}

\newcommand{\ok}{\Op{k}}

\newtheorem{theorem}{Theorem}
\newtheorem{Proposition}[theorem]{Proposition}
\newtheorem{Definition}[theorem]{Definition}

\usepackage{xcolor}

\begin{document}

\preprint{APS/123-QED}

\title{Quantum memory in spontaneous emission processes}
\author{Mei Yu}
\email{mei.yu@uni-siegen.de}
\affiliation{Naturwissenschaftlich-Technische Fakultät, 
Universität Siegen, Walter-Flex-Straße 3, 57068 Siegen, Germany}

\author{Ties-A. Ohst}
\email{ties-albrecht.ohst@uni-siegen.de}
\affiliation{Naturwissenschaftlich-Technische Fakultät, 
Universität Siegen, Walter-Flex-Straße 3, 57068 Siegen, Germany}

\author{Hai-Chau Nguyen}
\affiliation{Naturwissenschaftlich-Technische Fakultät, 
Universität Siegen, Walter-Flex-Straße 3, 57068 Siegen, Germany}

\author{Stefan Nimmrichter}
\affiliation{Naturwissenschaftlich-Technische Fakultät, 
Universität Siegen, Walter-Flex-Straße 3, 57068 Siegen, Germany}

\begin{abstract}
Quantum memory effects are essential in understanding and controlling open quantum systems, yet distinguishing them from classical memory remains challenging. We introduce a convex geometric framework to analyze quantum memory propagating in non-Markovian processes. We prove that classical memory between two time points is fundamentally bounded 
and introduce a robustness measure for quantum memory based on convex geometry. This admits an efficient experimental characterization by linear witnesses of quantum memory, bypassing full process tomography. 
We prove that any memory effects present in the spontaneous emission process of two- and three-level atomic systems are necessarily quantum, suggesting a pervasive role of quantum memory in quantum optics. Giant artificial atoms are discussed as a readily available test platform.
\end{abstract}

\maketitle

{\it Introduction.---}  
Decoherence due to the inevitable coupling to the environment is a cornerstone for our modern understanding of quantum theory and its applications~\cite{Schlosshauer2019}.
In the simplest approximation, decoherence processes are memory-less and can be described by Markovian quantum master equations~\cite{Breuer2002}, which are widely used, but do not completely capture the physics of the system-environment interaction during their evolution~\cite{Caruso2014, Breuer2016, deVega2017}.
Indeed, memory effects were found in various experiments with superconducting circuits~\cite{Gusafsson2014, Andersson2019, Kitzman2023}, quantum transport systems \cite{Engel2007, Collini2010, Panitchayangkoon2010}, spin networks \cite{Hanson2008, Niknam2021,Onizhuk2021}, optomechanical systems \cite{Groeblacher2015}, and photonic simulators \cite{Liu2011, Smirne2011, Hoeppe2012}.

In practice, non-Markovian effects of decohering environments may often be classical, in that the multi-time measurement statistics on an open system could be described by a classical stochastic process with memory \cite{Vacchini2011,Smirne2013,Smirne2018,Strasberg2019,Milz2020}; non-Markovian dephasing being a natural example \cite{Yao2007,Guarnieri2014,Chen2019}.  
In general, however, the memory that an environment relays back to the system can be coherent and genuinely quantum~\cite{Milz2020, Giarmatzi2021, Banacki2023}, and the distinction between quantum and classical memory is central to our understanding of open systems and our ability to control and simulate them in technological applications~\cite{Oreshkov2007,Bellomo2007, Huelga2012, Xue2012, Orieux2015, ohst_zhang_2024}. For example, non-Markovianity plays a role in stabilizing nonequilibrium states \cite{Grimsno2015, Ahlborn2004}, controlling quantum processors \cite{White2020}, increasing quantum channel capacities \cite{Bylicka2013}, and mitigating the noise that impedes quantum computation~\cite{Ahn2023,Puviani2025}.
\begin{figure}[t!]
    \centering
    \includegraphics[width=\linewidth]{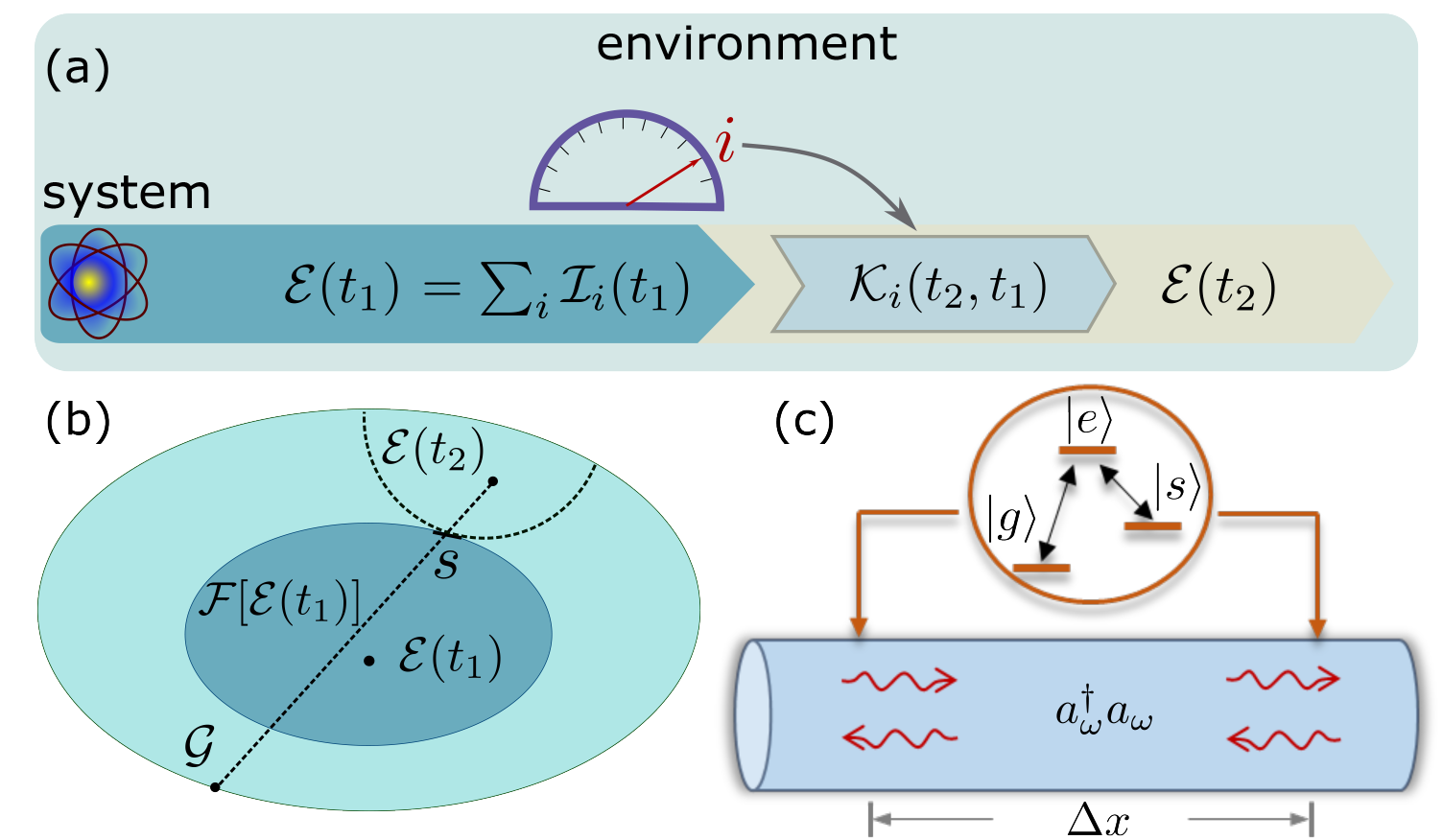}
    \caption{(a) Two-time quantum processes with classical memory admit a decomposition of the quantum channel $\mathcal{E}(t_1)$ describing the first time step into subchannels $\mathcal{I}_i(t_1)$ such that the classical information $i$ extracted by an environment controls the subsequent evolution $\mathcal{K}_{i}(t_2, t_1)$ in the second step, leading to $\mathcal{E}(t_2)$. 
    (b) The classical future $\mathcal{F}[\mathcal{E}(t_1)]$ is the convex subset of all those evolutions to later times $t_2$ that can be realized with classical memory. A two-time process exhibits quantum memory if the channel $\mathcal{E}(t_2)$ is detected outside of $\mathcal{F}[\mathcal{E}(t_1)]$. (c) Example of a physical process with quantum memory: spontaneous emission of a three-level giant atom into an acoustic waveguide via two distant coupling points.}
    \label{fig:robustness_GA_sketch}
\end{figure}

A rigorous description of memory effects in quantum dynamics was first developed 
in the process tensor framework \cite{Chiribella2009, Pollock2018,Giarmatzi2021,Taranto2024,taranto25}. 
This approach relies on process tomography of the dynamics at multiple time points, which poses
a significant challenge for practical measurements. 
Moreover, as analyzed in Ref.~\cite{Taranto2024}, deciding whether the memory is classical or quantum is a mathematically difficult problem.   
Only recently was it suggested to investigate the quantum memory propagating between two time points of the process without demanding coherent control over the intermediate evolution, eventually leading
to practical criteria for quantum memory of quantum processes of two-level systems~\cite{Bäcker24}, with limited later extension to higher dimensions~\cite{Bäcker2025}. 

In this work, we show that the quantum memory propagating between two time points of a non-Markovian process admits an elegant and systematic convex geometric formulation, illustrated in Fig.~\ref{fig:robustness_GA_sketch}(a) and (b),  shedding light on fundamental and practical aspects of  memory in quantum processes.
Fundamentally, we show that the amount of classical memory propagating between two time points is finite and bounded.
Practically, we formulate a robustness measure of quantum memory against noise in a natural way.
Moreover, the convex structure of processes with classical memory allows us to characterize quantum memory using linear witnesses, eliminating the need for process tomography in experimental demonstrations.
Finally, we point out that quantum memory is a prevalent feature in quantum optical systems: we prove for spontaneous emission processes of two- and three-level atoms into the vacuum that the memory there is necessarily of quantum nature. We illustrate this in the experimentally accessible setting of giant atoms~\cite{Gusafsson2014, Andersson2019}.

\noindent {\it Quantum memory in the dynamical process.---} 
Consider an open dynamics of a quantum system of dimension $d$ described by a one-parameter family of completely positive trace-preserving (CPTP) maps $\mathcal{E}({t})$, each mapping the initial state $\rho$ of the system at time $0$ to its state $\mathcal{E}({t})[\rho]$ at time $t$. 
Adopting the definition in~Ref.~\cite{Wolf2008}, we call such a dynamics memory-less if for every two time points $t_1$ and $t_2>t_1$, there exists a quantum channel $\mathcal{K}({t_2, t_1})$ such that
\begin{equation}
\label{eq:class_memory_dyn}
    \mathcal{E}({t_2}) = \mathcal{K}({t_2, t_1})  \mathcal{E}({t_1}).
\end{equation}
When this condition, known as `completely positive divisibility', is not satisfied, we speak of non-Markovian dynamics~\cite{Chruciski2022}. 

Non-Markovianity is an expression of memory effects in the dynamics. Intriguingly, the memory in a non-Markovian dynamical process can be of classical or of quantum nature. 
Following Ref.~\cite{Bäcker24}, {the memory propagating between two time points $t_1<t_2$ can be regarded as classical if the channel $\mathcal{E}({t_1})$ represents a nondestructive measurement, the outcomes of which influence the realization of $\mathcal{E}({t_2})$.} 
Formally, this means one can decompose $\mathcal{E}({t_1})$ into $n$ subchannels $\mathcal{I}_i({t_1})$ also known as instrument elements \cite{Busch2013-uk} such that  $\mathcal{E}({t_1}) = \sum_{i=1}^{n} \mathcal{I}_{i}({t_1})$. Each subchannel $\mathcal{I}_i({t_1})$ produces a classical outcome $i$ and a post-measurement state $\mathcal{I}_{i}({t_1})[\rho]$.
The classical information of the outcome $i$ is fed back into the system dynamics in the form of inducing a specific channel $\mathcal{K}_i({t_2, t_1})$, so that
\begin{equation}
\label{eq:class_memory_dynamics}
    \mathcal{E} ({t_2}) = \sum_{i=1}^{n} \mathcal{K}_i({t_2, t_1})  \mathcal{I}_i({t_1}).
\end{equation}
In particular, the Markovian dynamics~\eqref{eq:class_memory_dyn} is a special case of the dynamics \eqref{eq:class_memory_dynamics} with only a single outcome, $n=1$. 

\noindent {\it Quantum memory and its characterization by convex geometry.---}
Not every quantum dynamical process can be implemented by classical memory as in definition~\eqref{eq:class_memory_dynamics}, in which case the process is said to have quantum memory; see Appendix~\ref{app:class_future_examples} for instructive examples.  
Verifying the presence of quantum memory in a dynamical process was perceived as a difficult problem~\cite{Bäcker24,Bäcker2025,Taranto2024}.

Here, we observe that definition~\eqref{eq:class_memory_dynamics} can be posed as a convex optimisation program~\cite{boyd2004convex}, which directly implies several important consequences on fundamental and practical aspects of quantum memory. 
To this end, given the channel $\mathcal{E}(t_1)$ at time $t_1$, we define the \emph{classical future} $\mathcal{F}[\mathcal{E}(t_1)]$ to be the set of channels $\mathcal{G}$ that can be obtained via $\mathcal{G} = \sum_{i=1}^{n} \mathcal{K}_{i}(t_2, t_1) \mathcal{I}_{i}(t_1)$, where $\{ \mathcal{I}_{i}(t_1)\}_{i=1}^{n}$ is a subchannel decomposition of $\mathcal{E}(t_1)$ and $\mathcal{K}_{i}(t_2, t_1)$ are certain channels. By definition~\eqref{eq:class_memory_dynamics}, the dynamics $\mathcal{E}(t)$ can be implemented with classical memory if and only if for all two time points $t_1 < t_2$, one has  $\mathcal{E}(t_2) \in \mathcal{F}[\mathcal{E}(t_1)]$. 

It is clear from the definition that the classical future {consistently accommodates pre- and lossless post-processing at the earlier channel $\mathcal{E}(t_1)$, as well as post-processing after the later one $\mathcal{E} (t_2)$.} 
That is, given $\mathcal{E}(t_2) \in \mathcal{F}[\mathcal{E}(t_1)]$, then $\mathcal{D}\mathcal{E}(t_2) \mathcal{C} \in \mathcal{F}[\mathcal{U}\mathcal{E}(t_1)\mathcal{C}]$ for any unitary channel $\mathcal{U}$ and two general channels $\mathcal{C}, \mathcal{D}$;
see Appendix \ref{sec:invariance} for the details. 

A fundamental property of the classical future is that it is convex and compact; see Appendix \ref{app:finite_class_memory} for the proof. This simple mathematical property has important consequences, which we are to discuss below.

\noindent{\it Finite bounds of classical memory.---}
In general, the minimal value of $\ln (n)$ in the decomposition~\eqref{eq:class_memory_dynamics} can be thought of as the amount of actual classical memory propagating from the time point $t_1$ through the environment back to the system at time $t_2$.
This raises the question whether there exist processes on finite-dimensional quantum systems where an \emph{unbounded} amount of memory is feeding through the environment back onto the system at later time. 
We show that
the actual amount of classical memory propagating between two time points of a process in a system of dimension $d$ can 
be bounded as 
\begin{equation}
  \ln n \le \ln d^4 (d^4+1).  
\end{equation}
The proof is a consequence of the finite classical information content of quantum measurements~\cite{DAriano2011} and Carath\'eodory's principle of convex geometry~\cite{Carathodory1911} applied to the compact set $\mathcal{F}$ (see Appendix \ref{app:finite_class_memory} for the details). 

\noindent{\it Quantification of quantum memory.---}
The convexity of the classical future also allows us to quantify the robustness of quantum memory within a dynamical process by methods of convex programs~\cite{Vidal99, Steiner2003,Chitambar2019,Uola2019a}.
In particular, we consider
\begin{align}
s^* = \max_{s, \mathcal{G}} \; &s \label{eq:convex_combi_future}   \\    
\textnormal{s.t. }&s \mathcal{E}(t_2) + (1-s) \mathcal{G} \in \mathcal{F} [\mathcal{E}(t_1)], \,\, \mathcal{G} \textnormal{ is CPTP} \nonumber
\end{align}
The channel $\mathcal{E}(t_2)$ is in the classical future $\mathcal{F} [\mathcal{E}(t_1)]$ if and only if $s^* \geq 1$. 
Intuitively, if $s^{*}$ is small, there is no channel in the vicinity of $\mathcal{E}(t_2)$ which lies in the classical future of $\mathcal{E}(t_1)$, as sketched in Fig.~\ref{fig:robustness_GA_sketch}(b). Inspired by the generalized robustness of entanglement~\cite{Steiner2003}, one therefore can define $r^{*} = {1}/{s^{*}} - 1$ as the \textit{robustness of quantum memory} with respect to the the pair of channels $[\mathcal{E}(t_1), \mathcal{E}(t_2)]$.

\noindent {\it Witnesses of quantum memory.---} 
In order to detect quantum memory by  direct application of the convex program~\eqref{eq:convex_combi_future}, process tomography \cite{Chuang1997} of $\mathcal{E}(t_1)$ and $\mathcal{E}(t_2)$ is required, demanding a complete set of preparation and measurement settings. 
Surprisingly, the concept of the classical future can alleviate this problem, as it admits a strong convexity property:
the set of pairs of channels $(\mathcal{E},\mathcal{G})$ such that $\mathcal{G} \in \mathcal{F} [\mathcal{E}]$ is also convex.
As a direct consequence, this allows for the construction of linear witnesses to detect quantum memory optimized with respect to available data, also when full process tomography of $\mathcal{E}(t_1)$ and $\mathcal{E}(t_2)$ are not available. 
More precisely, if the information about channel $\mathcal{E}(t_\alpha)$ for $\alpha=1,2$ is probed by preparation of states $\{\rho_i^\alpha \}$ and measuring of observables  $\{O_{j}^{\alpha}\}$, one can design coefficients $w^\alpha_{kl}$ such that 
\begin{equation}
\label{eq:two_point_witness_condition}
    \sum_{\alpha,ij} w^\alpha_{ij} \mathrm{tr} \{ O_{j}^{\alpha} \mathcal{E} (t_\alpha) [\rho^\alpha_i] \} \geq 0
\end{equation}
whenever $\mathcal{E}(t_2) \in \mathcal{F}[\mathcal{E}(t_1)]$. A negative sum of expectation values on the left hand side of \eqref{eq:two_point_witness_condition} thus detects quantum memory also with limited experimental data when the set of prepared states and measured observables are not enough to reconstruct $\mathcal{E}(t_1)$ and $\mathcal{E}(t_2)$; see Appendix~\ref{app:memory_witness} and examples below. 
{Still, solving a general convex optimization problem such as \eqref{eq:convex_combi_future} and constructing an associated witness are computationally challenging. We alleviate them by the following relaxation into a semidefinite program (SDP).}

\noindent {\it Semidefinite program (SDP) relaxation.---}
{For a mathematical convenience, }
we employ the Choi-Jamiołkowski representation~\cite{Choi1972, Jamiokowski1972}: any quantum channel $\mathcal{G}$ over a system of dimension $d$ can be represented by the bipartite positive operator $G^{\rm AB} = \sum_{i,j=0}^{d-1} \ketbra{i}{j} \otimes \mathcal{G}(\ketbra{i}{j})$ which obeys $\tr_{B} [G^{\rm AB}] = \openone^{\rm A}$. 
In this representation, the channel $\mathcal{E}(t_2)$ is in the classical future of $\mathcal{E}(t_1)$ per definition~\eqref{eq:class_memory_dynamics} if its associated Choi operator $E^{\rm AB}(t_2)$ can be written as $E^{\rm AB}(t_2) = {\ }^{\rm DD'}\bra{\phi^+} X^{\rm ADD'B} \ket{\phi^+}^{\rm DD'}$, where 
\begin{equation}
\label{eq:class_mem_choi}
    X^{\rm ADD'B} =   \sum_{i}  I_{i}^{\rm AD}(t_1) \otimes K_{i}^{\rm D'B}(t_2, t_1), 
\end{equation}
with $I_{i}^{\rm AD}(t_1)$ and $K_{i}^{\rm D'B}(t_2, t_1)$ corresponding to the Choi operators
of $\mathcal{K}_{i}(t_2,t_1)$ and $\mathcal{I}_{i}(t_1)$ in~\eqref{eq:class_memory_dynamics}, respectively, and $\ket{\phi^+}_{\rm DD'} = \sum_{i=0}^{d-1} \ket{i} \otimes \ket{i}$ being the unnormalized maximally entangled state over ${\rm DD'}$; see Appendix \ref{app:memory_witness} for details. 
Remarkably, Eq.~\eqref{eq:class_mem_choi} resembles the definition of separable states, {which allows us to adopt various entanglement detection techniques \cite{Ghne2009} for quantum memory. In particular, we can relax the separability condition \eqref{eq:class_mem_choi} to that of having a positive partial transpose (PPT) with respect to the bipartition ${\rm (AD|D'B)}$. In addition, $X^{\rm ADD'B}$ must be positive and obey $\tr_{\rm B}[X^{\rm ADD'B}] = E(t_{1})^{\rm AD} \otimes {\mathds{1}^{\rm D'}}$ since the channels $\mathcal{K}_i$ are trace-preserving. 
Together these relaxed constraints turn \eqref{eq:convex_combi_future} into an efficiently computable SDP and the corresponding PPT witnesses of quantum memory can also be extracted; see Appendix \ref{app:memory_witness} and \ref{app:quantum memory_witness} for more details. In general, the relaxed SDP gives an upper bound to \eqref{eq:convex_combi_future}, which, in many considered cases, can also be verified to be exact by constructing matching lower bounds by
 a see-saw algorithm or more advanced techniques~\cite{ohst24}.}

\noindent {\it Quantum memory in the spontaneous emission process.---}
Let us now apply the developed formalism to the quantum optical prime example of spontaneous emission. This fundamental phenomenon arises due to the unavoidable interaction between the system and the vacuum fluctuations of a surrounding field. Depending on the spectral density of field modes, the emission process can be non-Markovian.
Given the textbook setting of a two-level system or a $\Lambda$-type three-level system with two non-degenerate allowed transitions, an (effectively) zero-temperature radiative environment, and a sufficiently weak system-environment coupling so as to justify the rotating wave approximation, we can show that the environment memory, once present, is genuinely quantum. 
That is, for any two times $t_1<t_2$, the spontaneous emission process is either memory-less or it has quantum memory.

We start with the instructive case of a two-level atom coupled to a bosonic vacuum field. In the rotating wave approximation, the time evolution of atom and field is generated by the Hamiltonian 
\begin{eqnarray}\label{TLS_gene_Hamiltonian_main}
    H &=& \hbar \omega_e \ketbra{e}{e} + \int_0^{\infty} d\omega \,  \hbar \omega a_{\omega}^{\dagger} a_{\omega}  \nonumber \\
    &&+\hbar \int_0^{\infty} d\omega \sqrt{\frac{\Gamma(\omega)}{4\pi}} \left( a_{\omega}^{\dagger} \ketbra{g}{e} + a_{\omega} \ketbra{e}{g} \right),
\end{eqnarray}
where $\ket{g}, \ket{e}$ respectively denote the ground state and the excited state of the atom, $\omega_e$ the transition frequency, $a_\omega$ the mode operators, and $\Gamma(\omega)$ an arbitrary spectral density of the atom-field coupling. 
Starting from the vacuum state and an arbitrary pure state of the atom, one can formally integrate the Schrödinger equation under the Hamiltonian~\eqref{TLS_gene_Hamiltonian_main} and express the evolution as $\mathcal{E}(t) \equiv \mathcal{C} [c(t)]$, given in the Kraus representation as
\begin{eqnarray} \label{eq:2TLS_qmap}
    \mathcal{C}[c]\rho &=& (1-|c|^2)\ketbra{g}{e} \rho\ketbra{e}{g} + \nonumber \\
    && + (\ketbra{g}{g} + c\ketbra{e}{e})\rho (\ketbra{g}{g} + c^*\ketbra{e}{e}),
\end{eqnarray}
$ {2\pi} c(t) = {i} \int_{-\infty}^{\infty} d\omega {e^{-i\omega t}}/{[\omega - \omega_e+i\tilde\chi(\omega)]}$ and $4\pi \tilde\chi(\omega) = \iint_0^{\infty} ds d\omega' \, \Gamma(\omega') e^{i(\omega-\omega')s} $. 
The channel $\mathcal{C}[c]$ has a rather remarkable property, $\mathcal{C}[c']\mathcal{C}[c] = \mathcal{C}[c'c]$, which implies that Eq.~\eqref{eq:class_memory_dyn} is solved by $\mathcal{K}(t_2,t_1) = \mathcal{C}[c(t_2)/c(t_1)]$. The transformation $\mathcal{K}(t_2, t_1)$ constitutes a valid channel whenever $|c(t_2)| \leq |c(t_1)|$, indicating no memory propagating through the environment back into the system between $t_1<t_2$. 
Noticing~\eqref{eq:2TLS_qmap} is a sparse channel, we prove a stronger result: the inequality $|c(t_2)|>|c(t_1)|>0$ is necessary and sufficient condition for the process to exhibit memory, which is also necessarily of quantum nature; see Appendix~\ref{app:dyn_two_level_sys} for the details.
Thus, the non-monotonicity of $c(t)$ fully characterizes the quantum memory of this spontaneous emission process.

We now turn our attention to the $\Lambda$-type three-level system with an additional metastable state $|s\ra$ of energy $\hbar\omega_s$ above the ground state Fig.~\ref{fig:robustness_GA_sketch}(c). The field can induce dipole-allowed transitions between $|e\rangle \leftrightarrow |g\rangle$ and $|e\rangle \leftrightarrow |s\rangle$ with frequencies $\omega_e$ and $\omega_s$, respectively, but direct transitions $|g\rangle \leftrightarrow |s\rangle$ are forbidden.  
The corresponding atom-field Hamiltonian in rotating wave approximation reads as
\begin{eqnarray} \label{3TLS_gene_Hamiltonian_main}
    H &=&  \hbar\omega_e \ketbra{e}{e} + \hbar\omega_s \ketbra{s}{s} + \int_0^{\infty} \!\! d\omega \, \hbar\omega  a_{\omega}^\da a_{\omega} \\
    &+& \hbar\int_0^{\infty} \!\! d\omega \left[ a_{\omega}^\da\frac{\sqrt{\Gamma_1(\omega)} \ketbra{g}{e} + \sqrt{\Gamma_2(\omega)} \ketbra{s}{e}}{\sqrt{4\pi}} + h.c. \right], \nonumber 
\end{eqnarray}
with the coupling spectra $\Gamma_1(\omega)$ and $\Gamma_2(\omega)$ characterizing the dipole-allowed  transitions.
As before, one can exactly solve the reduced atom dynamics starting from the vacuum state of the field, $\mathcal{E}(t) \equiv \mathcal{C} [d(t),G(t)]$, with
\begin{eqnarray}\label{eq:channel-3}
    \mathcal{C} [d,G]\rho &=& \la e|\rho|e\ra \left[G \ketbra{g}{g} + (1-G-|d|^2) \ketbra{s}{s} \right] + \ok \rho \ok^\da, \nonumber \\
    &&\ok = d \ketbra{e}{e} +e^{-i\omega_s t}\ketbra{s}{s} + \ketbra{g}{g}.
\end{eqnarray}
Here, ${4\pi} G(t) = \int_0^{\infty} d\omega \Gamma_1 (\omega) \left| \int_0^t dt'd (t')e^{i\omega t'} \right|^2$, 
${2\pi} d(t) = {i} \int_{-\infty}^{\infty} d \omega {e^{-i\omega t}}/[{\omega - \omega_e+i\tilde\chi_1(\omega) + i\tilde\chi_2(\omega-\omega_s)}]$, and
$4\pi\tilde{\chi}_j (\omega) = \iint_0^{\infty} ds d\omega' \Gamma_j(\omega') e^{i\left(\omega - \omega' \right) s}$; see Appendix \ref{app:dyn_three_level_sys} for the details.
Remarkably, the channels $\mathcal{C}[d,G]$ also adhere to a simple composition rule, $\mathcal{C}[d',G'] \mathcal{C}[d,G] = \mathcal{C}[d'd,G+|d|^2G']$, which implies that Eq.~\eqref{eq:class_memory_dyn} is solved by $\mathcal{K}(t_2,t_1) = \mathcal{C}[d(t_2)/d(t_1),(G(t_2)- G(t_1))/|d(t_1)|^2]$.  
Demanding $\mathcal{K}(t_2,t_1)$ be a valid channel leads to
\begin{align}
    G(t_2) &\geq G(t_1), \nonumber \\
    G(t_1) + \abs{d(t_1)}^2 &\geq G(t_2) + \abs{d(t_2)}^2.
    \label{eq:3-criteria}
\end{align}
When these inequalities both hold, the process is explicitly Markovian in the sense of definition~\eqref{eq:class_memory_dyn}.
Again, the channel~\eqref{eq:channel-3} is highly sparse, which allows us to prove a stronger result: a violation of either inequalities is necessary and sufficient for the presence of memory in the process, which is always of quantum nature (see Appendix~\ref{app:dyn_three_level_sys}).

\noindent {\it Application to giant artificial atoms.---} 
Being multilevel systems that couple to a bosonic field of surface acoustic waves (SAWs) through two or more spatially distant contacts, giant artificial atoms are ideal for studying memory effects in quantum processes~\cite{Gusafsson2014, Andersson2019,Kitzman2023}.
Consider such a two-level artificial giant atoms coupled to a one-dimensional SAW waveguide through two contacts separated by distance $\Delta x$ as depicted in Fig.~\ref{fig:robustness_GA_sketch}(c).  
Let $v_g$ be the velocity of sound, the revival time is then $\tau = \Delta x/v_g$. Given the spontaneous decay rate $\gamma$ and a spectrally flat atom-field coupling at each contact, we can describe this case by the Hamiltonian \eqref{TLS_gene_Hamiltonian_main} with spectral density $\Gamma(\omega) = 8 \gamma \cos^2{\left(\omega \tau /2\right)}$. Assuming the field is initially in the vacuum state, the resulting decay is given by the channel~\eqref{eq:2TLS_qmap} with the amplitude~\cite{Guo2017} (see also Appendix ~\ref{app:dyn_2GA})
\begin{equation}
    c(t) = \sum_{n=0}^{\infty} \Theta(t-n\tau) \frac{[-\gamma (t-n\tau)/2]^n}{n\,!}  e^{-(i\omega_e +\gamma/2)(t-n\tau)}. 
    \label{damp_fact}
\end{equation}
For simplicity, we have omitted the vacuum Lamb shift of the atomic resonance. The function~\eqref{damp_fact} captures a partially recurring backflow of the atom's initial excitation over multiples of $\tau$.
The excitation revivals are most pronounced whenever $\omega_e \tau$ is a multiple of $2\pi$ and the summands in \eqref{damp_fact} interfere constructively. 

Such non-monotonous excitation decay with partial revivals on time scales much longer than the spontaneous decay time was recently experimentally observed~\cite{Andersson2019}.
If the bath can be assumed to be in the vacuum and the dynamics~\eqref{eq:2TLS_qmap} is valid, this is not only the signature of memory effects in the decaying process~\cite{Andersson2019}, but also, as we have shown, a demonstration of its quantum nature.

\begin{figure}[t!]
    \includegraphics[width=0.49\textwidth]{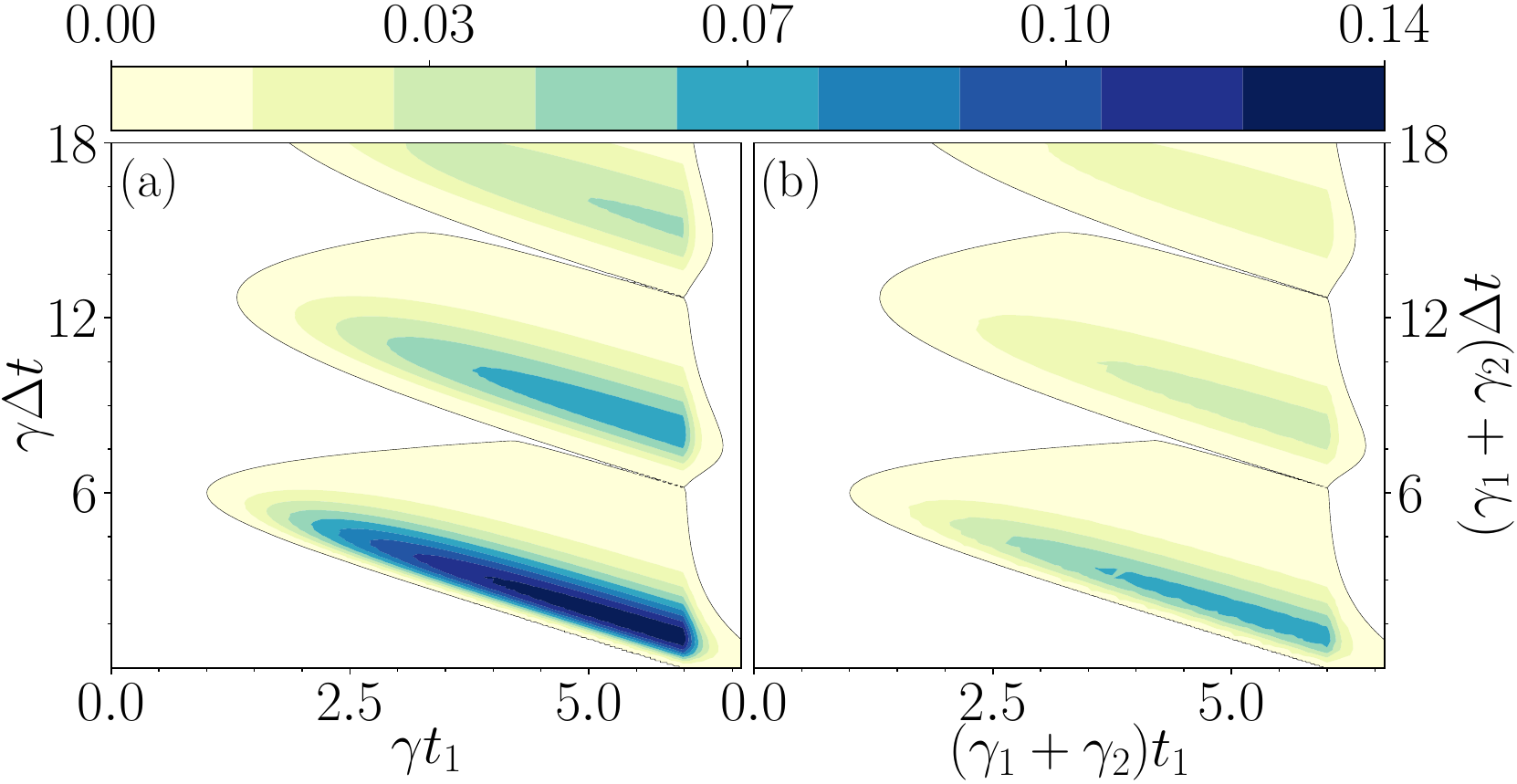}
    \caption{Quantum memory robustness between two times $(t_1, t_1 + \Delta t)$ in the spontaneous emission process of (a) a two-level giant atom with parameters $\omega_e \tau = 40 \pi,\, \gamma \tau = 12$ and (b) a comparable three-level giant atom with $\omega_s \tau = 20\pi$, $\omega_e \tau = 40\pi$, $\gamma_1 \tau = 4$, and $\gamma_2 \tau = 8$. } 
    \label{q_mem_robustness_GA}
\end{figure}

One might however argue, that due to experimental imperfection, the dynamics~\eqref{eq:2TLS_qmap} cannot be assumed to hold unconditionally. 
In fact, finite temperature effects can smear out the backflow memory and its quantum nature.
In this case, further characterization of the process is required to conclude on the memory effect as well as its quantum nature.

In the case where process tomography can be carried out for two time points $t_1$ and $t_2$, the SDP~\eqref{eq:convex_combi_future} can be used to verify the presence of quantum memory. To have an idea about the robustness of the quantum memory effect against noise, we use Eq.~\eqref{eq:convex_combi_future} to evaluate the robustness $r^*$ for the ideal process~\eqref{eq:2TLS_qmap}. In Fig.~\ref{q_mem_robustness_GA}(a), the robustness is plotted as a function of the two time points $t_1$ and $t_2 = t_1 + \Delta t$, for an exemplary setting with an appreciable revival time. The black lines mark the boundaries between Markovian processes (unshaded region, $r^*\leq 0$) and processes with quantum memory of varying robustness (shaded). 
The robustness reaches its global maximum, and thus the impact of coherent information backflow is the strongest, when $t_1$ and $t_2$ are shortly before and after the first revival time $\tau$. At fixed $t_1$, the robustness periodically decays and revives in parts with growing $t_2$, as the system loses and regains information. 

Process tomography at two time points generally requires all $32$ different expectation values, $\Lambda_{ij}^{\alpha} := \tr(\mathcal{E}(t_\alpha)[\sigma_{i}] \sigma_{j})$, with $i,j = 0,1,2,3$ labeling the Pauli matrices (including the identity) and $\alpha=1,2$ the two time points. 
To illustrate the method of witnesses using limited information, we construct a witness requiring only $12$ different values of $\Lambda_{ij}^\alpha$ for the most robust time points, $\gamma t_1 \approx  5.9$ and $\gamma t_2 \approx 7$ seen in Figure~\ref{q_mem_robustness_GA}(a) with the detailed description in Appendix~\ref{app:dyn_two_level_sys}. 

As an exemplary three-level study, we consider the $\Lambda$-type giant atom~\cite{Du2022, Vadiraj2021, Du2021} with configuration sketched in Fig.~\ref{fig:robustness_GA_sketch}(c), where
the allowed transitions $|e\rangle \leftrightarrow |g\rangle$ and $|e\rangle \leftrightarrow |s\rangle$ are both coupled via two contacts to the SAW waveguide, with delay time $\tau$ and single-contact decay rates $\gamma_1$ and $\gamma_2$, respectively. 
The effective Hamiltonian is given by Eq.~\eqref{3TLS_gene_Hamiltonian_main}, with $\Gamma_{i=1,2}(\omega) = 8 \gamma_i \cos^2{\left(\omega \tau /2\right)}$. The resulting decay function $d(t)$ takes a form similar to \eqref{damp_fact}, which fully determines the time evolution if we neglect the Lamb shift; see  Appendix~\ref{app:dyn_three_level_sys}. 

We plot the quantum memory robustness against $t_1$ and $t_2$ in Fig.~\ref{q_mem_robustness_GA}(b) for a three-level setting with the same overall excited state decay rate as in the two-level case (a), $\gamma_1 + \gamma_2 = \gamma$. The memory properties are periodic in both transition frequencies $\omega_e,\omega_s$ with period $2\pi/\tau$ and the most pronounced when both $\omega_e\tau$ and $\omega_s\tau$ are multiples of $2\pi$.
The regimes of quantum memory (shaded regions) then coincide almost exactly with the two-level case, but the overall robustness is reduced by about a factor two. 

To illustrate the adaptability of the witnesses to the accessible measurements, we restrict to Pauli measurements and preparation of Pauli eigenstates using only levels with direct dipole coupling, $|e\rangle \leftrightarrow |g\rangle$ and $|e\rangle \leftrightarrow |s\rangle$. We illustrate the construction of witnesses with respect to these measurement settings to detect quantum memory at the optimally robust time points given by $(\gamma_1 + \gamma_2)t_1 \approx 6$ and $(\gamma_1 + \gamma_2)t_2 \approx 6.92$, see Figure~\ref{q_mem_robustness_GA}(b); see Appendix \ref{app:dyn_three_level_sys} for the details.

\noindent {\it Conclusion.---}
The presented convex geometry framework serves as a basis for 
deepening our understanding of memory in quantum dynamics.
Knowing that the classical memory transferred between two time points of a process is fundamentally bounded, one could further distinguish processes based on different amounts of memory. Applying our criteria to consecutive infinitesimal time steps could establish a link between the presence of quantum memory and the exact form of a non-Markovian master equation.
Our approach could also bring deeper understanding of quantum memory in the process tensor formalism.
Finally, an extension to continuous variables could reveal or facilitate memory effects in the physics of quantum optical or mechanical systems.

\acknowledgments
M.Y. and T.A.O. contributed equally.
We thank Charlotte Bäcker, Leonardo S.V.~Santos, and  Walter Strunz,  for fruitful discussion.
This work is supported by 
the Deutsche Forschungsgemeinschaft (DFG, German Research Foundation, project numbers 447948357 and 440958198), 
the Sino-German Center for Research Promotion (Project M-0294), 
the German Ministry of Education and Research (Project QuKuK, BMBF Grant No. 16KIS1618K). H.C.N. is further supported by the EIN~Quantum~NRW.


%

\onecolumngrid
\appendix

\section{Elementary examples of classical and quantum memory in dynamical processes}
\label{app:class_future_examples}

To illustrate the concept of classical memory and classical futures as a proper generalization of Markovian dynamics, one can consider the dephasing qubit dynamics $\mathcal{D}_{\alpha}(t) \in C_2$, where $C_d$ denotes the set of channels on a $d$-level system, defined by 
\begin{alignat}{2}
    \mathcal{D}_{\alpha}(t)[\ketbra{0}{0}] &= \ketbra{0}{0}  &\mathcal{D}_{\alpha}(t)[\ketbra{1}{1}] &= \ketbra{1}{1} \\
        \mathcal{D}_{\alpha}(t)[\ketbra{0}{1}] &= \alpha(t) \ketbra{0}{1} \quad &\mathcal{D}_{\alpha}(t)[\ketbra{1}{0}] &= \alpha(t)^{*} \ketbra{1}{0}.
\end{alignat}
Here, $\alpha(t)$ is any function such that $\abs{\alpha(t)} \leq 1$ for all $t$, guaranteeing that $\mathcal{D}_{\alpha}(t)$ is completely positive for all times $t$. In terms of memory, the process $D_{\alpha}(t)$ has the following characterization.

\begin{Proposition}
    \label{prop:dephasing_dynamics}
    The two point qubit dephasing dynamics $\mathcal{D}_{\alpha}(t): \mathbb{R} \rightarrow C_2$ is Markovian if and only if 
    \begin{equation}
        \abs{\alpha(t_2)} \leq \abs{\alpha(t_1)}
    \end{equation}
    for all times $t_1$ and $t_2$ such that $t_{1} \leq t_{2}$.
    On the other hand, the dynamics can be implemented with classical memory for all values of $\alpha(t)$, i.e., $\mathcal{D}_{\alpha}(t_2) \in \mathcal{F}[\mathcal{D}_{\alpha}(t_1)]$, for all functions $\alpha$ such that $\abs{\alpha(t)} \leq 1$ and for all time points $t_1$ and $t_2$. Moreover, the size $n$ of the subchannel decomposition of $\mathcal{D}_{\alpha}(t_1)$ can be chosen as $n=2$.
\end{Proposition}
\begin{proof}
    For the Markovianity, consider the map $\mathcal{K}_{\alpha}(t_2, t_1)$ that satisfies $\mathcal{D}_{\alpha}(t_2) = \mathcal{K}_{\alpha}(t_2, t_1) \circ \mathcal{D}_{\alpha}(t_1)$. It is straightforward to see that $\mathcal{K}_{\alpha}(t_2, t_1)$ must be given by the action 
    \begin{alignat}{2}
    \mathcal{K}_{\alpha}(t_2, t_1)[\ketbra{0}{0}] &= \ketbra{0}{0} 
    \quad &\mathcal{K}_{\alpha}(t_2, t_1)[\ketbra{1}{1}] &= \ketbra{1}{1} \\
    \mathcal{K}_{\alpha}(t_2, t_1)[\ketbra{0}{1}] &= (\alpha(t_2)/\alpha(t_1)) \ketbra{0}{1} \quad &
    \mathcal{K}_{\alpha}(t_2, t_1)[\ketbra{1}{0}] &= (\alpha(t_2)/\alpha(t_1))^{*} \ketbra{1}{0}
\end{alignat}
such that $\mathcal{K}_{\alpha}(t_2, t_1)$ is completely positive if and only if $\abs{\alpha(t_2)} \leq \abs{\alpha(t_1)}$.

We now show that it is always possible to decompose $\mathcal{D}_{\alpha}(t_2)$ as 
\begin{equation}
\label{eq:dephasing_classical_deco}
    \mathcal{D}_{\alpha}(t_2) = \sum_{i=1}^{2} \mathcal{K}_{\alpha i}(t_2, t_1) \circ \mathcal{I}_{\alpha i}(t_1)
\end{equation}
    with channels $\mathcal{K}_{\alpha i}(t_2, t_1)$ and completely positive maps $\mathcal{I}_{\alpha i}(t_1)$ such that $\mathcal{D}_{\alpha}(t_1) = \mathcal{I}_{\alpha 1}(t_1) + \mathcal{I}_{\alpha 2}(t_1)$. The Choi operators of the channels and instrument elements can on their support (spanned by $\ket{00}$ and $\ket{11}$) explicitly constructed as 
    \begin{equation}
        I_{\alpha 1}(t_1) = \begin{pmatrix}
            \frac{1}{2} & \gamma \\
            \gamma^{*} & \frac{1}{2}
        \end{pmatrix}
        \quad\quad\quad
        I_{\alpha 2}(t_1) = \begin{pmatrix}
            \frac{1}{2} & \alpha(t_1) - \gamma \\
            \alpha(t_1)^{*} - \gamma^{*} & \frac{1}{2}
        \end{pmatrix}
    \end{equation}
    and 
    \begin{equation}
        K_{\alpha 1}(t_2, t_1) = \begin{pmatrix}
            1 & \frac{\delta}{\gamma} \\
            \frac{\delta^{*}}{\gamma^{*}} & 1
        \end{pmatrix}
        \quad\quad\quad
        K_{\alpha 2}(t_2, t_1) = \begin{pmatrix}
            1 & \frac{\alpha(t_{2}) - \delta}{\alpha(t_{1}) - \gamma} \\
            \frac{\alpha(t_{2})^{*} - \delta^{*}}{\alpha(t_{1})^{*} - \gamma^{*}} & 1
        \end{pmatrix}
    \end{equation}
    with some complex numbers $\gamma$ and $\delta$. 
    It can be checked that the associated subchannels and channels together satisfy Eq.~\eqref{eq:dephasing_classical_deco} and, if the complex numbers $\gamma$ and $\delta$ are chosen to obey $\abs{\gamma} = \abs{\delta} = \frac{1}{2}$ and $\abs{\alpha(t_1) - \gamma} = \abs{\alpha(t_2) - \delta} = \frac{1}{2}$, it is ensured that all the maps are completely positive. A simple geometric argument reveals that such numbers $\gamma$ and $\delta$ always exist for all $\alpha(t_1)$ and $\alpha(t_2)$ on the complex unit disk, see Figure \ref{fig:dephasing_proof}. 
 \end{proof}
 
The dephasing dynamics $\mathcal{D}_{\alpha}(t)$ is an example for a process that shows non-Markovianity that solely arises from classical memory effects. In contrast to that, a simple example which shows the presence of genuine quantum memory is given by the qubit dynamics $\mathcal{E}_{J}(t)$ defined as 
 \begin{equation}
     \mathcal{E}_{J}(t)[\rho] = \Tr_{\rm E}[e^{itH_{J}} (\rho \otimes \ketbra{1}{1}_{\rm E}) e^{-itH_{J}}]
 \end{equation}
where $\rm{E}$ denotes the environment two-level system and where the Hamiltonian $H_{J}$ is
\begin{equation}
    2H_{J} / \hbar = -J_{x} \sigma_{x} \otimes \sigma_{x} - J_{y} \sigma_{y} \otimes \sigma_{y} - J_{z} \sigma_{z} \otimes \sigma_{z}.
\end{equation}
with some coupling rates $J_x, J_y$ and $J_z$. In the Figures \ref{fig:heisenberg_qm} and \ref{fig:heisenberg_markov}, the robustness of quantum memory and non-Markovianity are plotted for $(J_x, J_y, J_z) = (-1, -2, -3) \frac{1}{s}$. The computation of those robustness is carried out with semidefinite programs as explained below in Section \ref{app:memory_witness}. In comparison to the dephasing dynamics, one can see that for some values of time pairs $(t, t + \Delta t)$, the two point dynamics is classical but still non-Markovian. 

\section{Invariance of the classical future under pre- and postprocessing operations}
\label{sec:invariance}
\begin{Proposition}
    Let $\mathcal{E}, \mathcal{G} \in C_d$ be two quantum channels and let $\mathcal{G}$ be in the classical future of $\mathcal{E}$, i.e., $\mathcal{G} \in \mathcal{F}[\mathcal{E}]$. Furthermore, let $\mathcal{C}, \mathcal{D} \in C_d$ be some channels and let $\mathcal{U} \in U_d$ be a unitary channel. Then, it holds that $\mathcal{D} \circ \mathcal{G} \circ \mathcal{C}$ is in the classical future of $\mathcal{U} \circ \mathcal{E} \circ \mathcal{C}$, i.e. $\mathcal{D} \circ \mathcal{G} \circ \mathcal{C} \in \mathcal{F}[\mathcal{U} \circ \mathcal{E} \circ \mathcal{C}]$. 
 
\end{Proposition}
\begin{proof}
    We assume that $\mathcal{G} \in \mathcal{F}[\mathcal{E}]$. Then by definition of the classical future, there exists an $n \in \mathbb{N}$, quantum channels $\{\mathcal{K}_{i}\}_{i=1}^{n}$ and a subchannel decomposition $\{\mathcal{I}_{i}\}_{i=1}^{n}$ of $\mathcal{E}$ such that $\mathcal{G} = \sum_{i=1}^{n} \mathcal{K}_{i} \circ \mathcal{I}_i$.
    For the channel $\mathcal{D} \circ \mathcal{G} \circ \mathcal{C}$, one calculates
    \begin{align*}
        \mathcal{D} \circ \mathcal{G} \circ \mathcal{C} = \sum_{i=1}^{n} \mathcal{D} \circ \mathcal{K}_{i} \circ \mathcal{I}_i \circ \mathcal{C} 
        = \sum_{i=1}^{n} \underbrace{\mathcal{D} \circ \mathcal{K}_{i} \circ \mathcal{U}^{-1}}_{=:\widetilde{\mathcal{K}}_{i}} \circ \underbrace{\mathcal{U} \circ \mathcal{I}_i \circ \mathcal{C}}_{=:\widetilde{\mathcal{I}}_{i}} 
        = \sum_{i=1}^{n} \widetilde{\mathcal{K}}_{i} \circ \widetilde{\mathcal{I}}_{i}
    \end{align*}
    where $\widetilde{\mathcal{K}}_{i}$ are channels and $\{\widetilde{\mathcal{I}}_{i}\}_{i=1}^{n}$ can be verified as a subchannel decomposition of $\mathcal{U} \circ \mathcal{E} \circ \mathcal{C}$. 
\end{proof}

\begin{figure}
\centering
\subfloat[\label{fig:dephasing_proof}]{\includegraphics[width=0.22\linewidth]{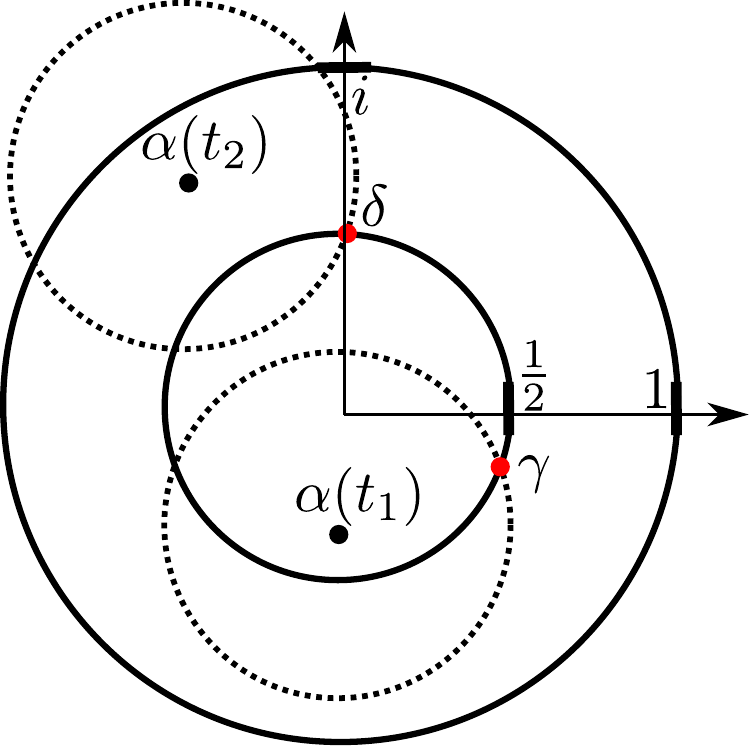}}
\hspace{0.02\linewidth}
\subfloat[\label{fig:heisenberg_qm}]{\includegraphics[width=0.36\linewidth]{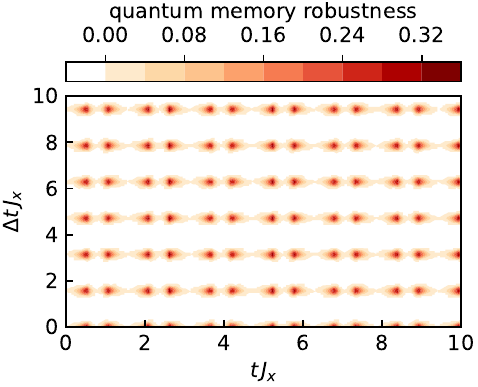}}
\hspace{0.02\linewidth}
\subfloat[\label{fig:heisenberg_markov}]{\includegraphics[width=0.36\linewidth]{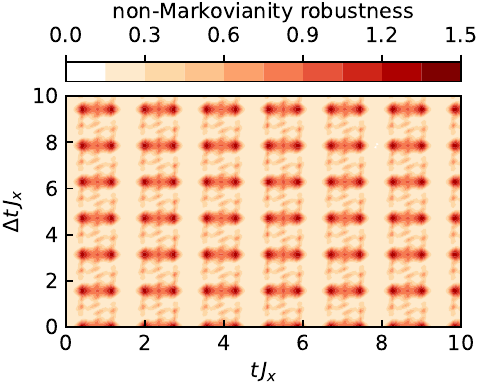}}
    \caption{(a) Geometric construction in the proof of Proposition \ref{prop:dephasing_dynamics} on the classical description of the qubit dephasing dynamics. The complex numbers $\alpha(t_1)$ and $\alpha(t_2)$ are confined to the unit disk. For every complex number $z$ on the unit disk, there exists at least one number $w$ such that $\abs{w}=\frac{1}{2}$ as well as $\abs{z-w} = \frac{1}{2}$. The number $w$ is an intersection point of the circle with radius $\frac{1}{2}$ around $0$ and the circle with radius $\frac{1}{2}$ around $z$. For $\alpha(t_1)$, this number is denoted as $\gamma$ and it is denoted as $\delta$ for $\alpha(t_2)$. (b) Quantum memory robustness of the two-qubit Heisenberg model at different time pairs $(t_{1}, t_{1} + \Delta t)$. (c) Non-Markovianity robustness of the two-qubit Heisenberg model at different time pairs $(t_{1}, t_{1} + \Delta t)$.}
    \label{fig:dephasing_class_geometry}
\end{figure}
\section{Elementary properties of the classical future and the fundamental bound of classical memory}
\label{app:finite_class_memory}
 
Here we give the proofs of elementary properties of the classical future and eventually arrive at the proof for the statement in the main text, that classical memory in dynamical processes over finite dimensional systems is finite and bounded.

We first notice that the classical future, as it has been defined in the main text, can be decomposed as 
\begin{equation}
    \mathcal{F}[\mathcal{E}] = \bigcup_{n \in \mathbb{N}} \mathcal{F}_{n}[\mathcal{E}]
\end{equation}
where $\mathcal{F}_{n}[\mathcal{E}]$ is the set defined by
\begin{equation}
    \mathcal{F}_{n}[\mathcal{E}] = \{\mathcal{G} \in C_{d}: \mathcal{G} = \sum_{i=1}^{n} \mathcal{K}_{i} \circ \mathcal{I}_i, \mathcal{K}_{i} \in C_d, \{\mathcal{I}_{i}\}_{i=1}^{n} \in I_{n}[\mathcal{E}]\}
    \label{eq:finite-classical-future}
\end{equation}
and $I_{n}[\mathcal{E}]$ denotes the set of all subchannel decompositions of $\mathcal{E}$ of size $n$.
In words, the set $\mathcal{F}_{n}[\mathcal{E}]$ is the subset of the classical future of $\mathcal{E}$ that may be reached with classical memory bounded by the size $n$ of the subchannel decomposition of $\mathcal{E}$.
It is clear that
\begin{equation}
    \mathcal{F}_0 [\mathcal{E}] \subseteq \mathcal{F}_1 [\mathcal{E}] \subseteq \mathcal{F}_2 [\mathcal{E}]
    \subseteq
    \cdots
    \subseteq
    \mathcal{F}_n [\mathcal{E}]
    \subseteq
    \cdots 
    \label{eq:chain}
\end{equation}
as it is always possible to pad the subchannel decompositions with additional zero maps. 
Although $\mathcal{F}_n[\mathcal{E}]$ themselves may not be convex, one can show that their union $\mathcal{F}[\mathcal{E}]$ is.
\begin{theorem}
     The classical future $\mathcal{F}[\mathcal{E}]$ of a channel $\mathcal{E} \in C_d$ is a convex set. In fact, $\mathrm{conv}(F_{n_1}[\mathcal{E}] \cup F_{n_2}[\mathcal{E}]) \subseteq F_{n_1+n_2}[\mathcal{E}]$.  
\end{theorem}
 \begin{proof}
     Let $\mathcal{G}_1$ and $\mathcal{G}_2$ be quantum channels that are in the classical future $\mathcal{F}[\mathcal{E}]$ of the channel $\mathcal{E}$. Then, these channels admit decompositions of the form 
    \begin{alignat}{2}
        &\mathcal{G}_1 = \sum_{i=1}^{n_1} \mathcal{K}_{1i} \circ \mathcal{I}_{1i} , \quad\quad\quad\quad
        &&\mathcal{G}_2 = \sum_{j=1}^{n_2} \mathcal{K}_{2j} \circ \mathcal{I}_{2j}
    \end{alignat}
    where $\mathcal{K}_{lj}$ are quantum channels for all $l,j$ and $\mathcal{I}_{lj}$ are completely positive maps such that $\sum_{i=1}^{n_{1}} \mathcal{I}_{1i} = \mathcal{E}_{1}$ as well as $\sum_{j=1}^{n_{2}} \mathcal{I}_{2j} = \mathcal{E}_{2}$ and $n_{1}, n_{2} \in \mathbb{N}$. Let now $\mathcal{G}$ be a convex combination, i.e., $\mathcal{G}= \lambda \mathcal{G}_1 + (1-\lambda) \mathcal{G}_2$ for some $\lambda \in [0,1]$. Then we get
    \begin{equation}
        \mathcal{G} = \sum_{i=1}^{n_1} \mathcal{K}_{1i} \circ \lambda\mathcal{I}_{1i} + \sum_{j=1}^{n_2} \mathcal{K}_{2j} \circ (1-\lambda)\mathcal{I}_{2j}
    \end{equation}
    which shows that $\mathcal{G}$ is in the classical future of $\mathcal{E}$, since $\{\lambda\mathcal{I}_{1i}\}_{i=1}^{n_{1}} \cup \{(1-\lambda)\mathcal{I}_{2j}\}_{j=1}^{n_{2}}$ is a subchannel decomposition of $\mathcal{E}$.
 \end{proof}

The finite memory bound we found may be formulated as follows
\begin{theorem}  \label{thm:finite_class_memory}
    Let $\mathcal{E}, \mathcal{G}$ be channels over a quantum system of dimension $d$. If $\mathcal{G} \in \mathcal{F}[\mathcal{E}]$ then $\mathcal{G} \in \mathcal{F}_{n}[\mathcal{E}]$ for some $n \le d^4(d^4 + 1)$. In other words, the chain~\eqref{eq:chain} saturates as equality for $n \le d^4(d^4 + 1)$. 
 \end{theorem}

To arrive at the proof, we proceed with elementary properties of the memory bounded classical futures $\mathcal{F}_n(\mathcal{E})$.
\begin{Proposition}
    \label{prop:Fn_compact}
     For any natural number $n$, the memory bounded classical future $\mathcal{F}_{n}[\mathcal{E}]$
     of any channel $\mathcal{E}$ over a  quantum system of finite dimension $d$ is a compact set. 
 \end{Proposition}
 \begin{proof}
Intuitively, this is the case because $\mathcal{F}_n[\mathcal{E}]$ is continuously parametrized by 
parameters with compact domains.
Formally, let $C_d$ denote the set of channels, and $I_n[\mathcal{E}]$ denote the set of subchannel decompositions of $\mathcal{E}$ of size $n$. Then $\mathcal{F}_{n}[\mathcal{E}]$ is the image of the continuous map $F_{n}: \underbrace{C_{d} \times \dots \times C_{d}}_{n \textnormal{ times}} \times I_{n}[\mathcal{E}] \rightarrow C_{d}$ that is defined as 
\begin{equation}
       F_{n}(\mathcal{K}_{1}, \dots, \mathcal{K}_{n}, \{\mathcal{I}_{i}\}_{i=1}^{n}) := \sum_{i=1}^{n} \mathcal{K}_{i} \circ \mathcal{I}_i.  
    \end{equation}
As the domain $C_{d} \times \dots \times C_{d} \times I_{n}[\mathcal{E}]$ is compact, so is its continuous image $\mathcal{F}_{n}[\mathcal{E}]$.
 \end{proof}

The concept of extremality plays a key role in the proof of Theorem \ref{thm:finite_class_memory} below. In general, a point $k \in K$ in a convex set $K$ is said to be extremal if any decomposition $k = pk_1 + (1-p)k_2$ with $p \in (0,1)$ and $k_1, k_2 \in K$ implies that $k_1 = k_2 = k$. We will denote the set of extremal points of a convex set $K$ as $\textnormal{ext}(K)$. 

In particular, the set of subchannel decompositions $I_{n}[\mathcal{E}]$ is a convex set where the convex combination
of two decompositions $\{\mathcal{I}_{1i}\}_{i=1}^{n}$ and $\{\mathcal{I}_{2i}\}_{i=1}^{n}$ with mixing parameter $p \in [0,1]$ is given by the subchannel decomposition $\{p\mathcal{I}_{1i} + (1-p)\mathcal{I}_{2i}\}_{i=1}^{n}$. 
We are to characterize the extremal subchannel decompositions in the following Proposition, the proof of which is mainly inspired by the techniques discussed in Ref.~\cite{DAriano2011}.    
 \begin{Proposition}
     \label{prop:extremal_instrument}
     Let $\{\mathcal{I}_{i}\}_{i=1}^{n} \in \textnormal{ext}(I_{n}[\mathcal{E}])$, $\mathcal{I}_{i} \neq 0$, be an extremal subchannel decomposition of a channel $\mathcal{E} \in C_d$, then $\{\mathcal{I}_{i}\}_{i=1}^{n}$ are linearly independent. In particular, as the dimension of the space of subchannels is $d^4$, at most $d^4$ subchannels in an extremal decomposition $\{\mathcal{I}_{i}\}_{i=1}^{n}$ can be non-zero. 
 \end{Proposition}
 \begin{proof}
     The statement can be proven by contradiction. We assume that the maps $\mathcal{I}_{i}$ are non-zero and not linearly independent. Therefore, there exists a non-zero vector $\alpha \in \mathbb{R}^{n}$ such that $\sum_{i=1}^{n} \alpha_{i} \mathcal{I}_{i} = 0$. 
     We can now define the maps 
    \begin{align}
       \mathcal{I}_{1i,\epsilon} = (1+\epsilon \alpha_{i}) \mathcal{I}_i \quad\quad\quad\quad
        \mathcal{I}_{2i, \epsilon} = (1-\epsilon \alpha_{i}) \mathcal{I}_i 
    \end{align}
    for some $\epsilon > 0$ such that $\mathcal{I}_{1i, \epsilon}$ and $\mathcal{I}_{2i,\epsilon}$ are completely positive for all $i=1,\dots,n$. Such an $\epsilon$ exists and can be chosen as the value
    \begin{equation}
        \epsilon = \max\{r: 1-|r\alpha_{i}| \geq 0 \textnormal{ for all } i =1, 2, \ldots, n \}.
    \end{equation}
    One can see that $\{\mathcal{I}_{1i,\epsilon}\}_{i=1}^{n}$ and $\{\mathcal{I}_{2i,\epsilon}\}_{i=1}^{n}$ are distinct subchannel decompositions that sum up to $\mathcal{E}$ and that the subchannels in $\{\mathcal{I}_{i}\}_{i=1}^{n}$ can be decomposed as $\mathcal{I}_{i} = \frac{1}{2}(\mathcal{I}_{1i,\epsilon} + \mathcal{I}_{2i,\epsilon})$ for all $i=1,\dots,n$. This means that $\{\mathcal{I}_{i}\}_{i=1}^{n}$ is not extremal implying a contradiction.
 \end{proof}

Let us now notice that, since $\mathcal{F}_n[\mathcal{E}]$ is compact as shown in Proposition~\ref{prop:Fn_compact}, by Carathéodory's Theorem \cite{Carathodory1911}, it holds that we can decompose every element of the convex hull $\textnormal{conv}(\mathcal{F}_{n}[\mathcal{E}])) \subseteq \mathcal{F}_{n}[\mathcal{E}]$ by finitely many extreme points in $\textnormal{ext}(\textnormal{conv}(\mathcal{F}_{n}[\mathcal{E}]))$.
The next Proposition~\ref{prop:extreme_points_cf} allows us to 
characterize the extreme points of the convex hull of $\mathcal{F}_n[\mathcal{E}]$. 
\begin{Proposition}
     \label{prop:extreme_points_cf}
     Let $\mathcal{G} \in \textnormal{ext}(\textnormal{conv}(\mathcal{F}_{n}[\mathcal{E}])) \subseteq \mathcal{F}_{n}[\mathcal{E}]$ be an extremal element of the memory bounded classical future of $\mathcal{E}$. Then there is an extremal subchannel decomposition $\{\mathcal{I}_{\rm ext,i}\}_{i=1}^{n} \in \textnormal{ext}(I_{n}[\mathcal{E}])$ such that
     \begin{equation}
         \mathcal{G} = \sum_{i=1}^{n} \mathcal{K}_{i} \circ \mathcal{I}_{\rm ext, i}. 
     \end{equation}
 \end{Proposition}
\begin{proof}
    As $\mathcal{G} \in 
 \mathcal{F}_{n}[\mathcal{E}]$, it can be decomposed as $\mathcal{G} = \sum_{i=1}^{n} \mathcal{K}_{i} \circ \mathcal{I}_{i}$
 for some subchannel decomposition $\{\mathcal{I}_{i}\}_{i=1}^{n}$ of $\mathcal{E}$. We now decompose $\mathcal{I}_{i}$ into a convex combination of extremal subchannel decompositions as
 \begin{equation}
\mathcal{I}_{i} =  \sum_{t} q_{t} \,{\mathcal{I}_{\rm{ext},i}(t)}
\end{equation}
for certain non-zero probability weights $q_t$ and associated extremal subchannel decompositions $\{\mathcal{I}_{\rm{ext},i}(t)\}_{i=1}^{n}  \subset \textnormal{ext}(I_{n}[\mathcal{E}])$ for each $t$.
It then follows that
\begin{align}
\mathcal{G} 
&= \sum_{t} q_{t} \underbrace{\left[\sum_{i=1}^{n} \mathcal{K}_{i} \circ \mathcal{I}_{\rm{ext}, i}(t)\right]}_{ \in \mathcal{F}_{n}[\mathcal{E}]}.
\end{align}
By extremality of $\mathcal{G}$, it must be true that $\mathcal{G} = \sum_{i=1}^{n} \mathcal{K}_{i} \circ \mathcal{I}_{\rm{ext}, i}(t)$ whenever $q_t > 0$, which completes the proof. 
\end{proof}
We are now ready prove the general bound for the size of the classical memory.

 \begin{proof} [Proof of Theorem~\ref{thm:finite_class_memory}]
    First, as $\mathcal{G} \in \mathcal{F}[\mathcal{E}]$, there is an $m \in \mathbb{N}$ such that $\mathcal{G} \in \mathcal{F}_{m}[\mathcal{E}]$. 
    By Proposition \ref{prop:Fn_compact}, we know that $\textnormal{conv}(\mathcal{F}_{m}[\mathcal{E}])$ is a compact and convex set and by the Krein-Milman Theorem \cite{Krein1940}, $\textnormal{conv}(\mathcal{F}_{m}[\mathcal{E}])$ is the convex hull of its extreme points. As $\textnormal{conv}(\mathcal{F}_{m}[\mathcal{E}])$ is embedded in a $d^4$-dimensional real vector space, Carathéodory's Theorem \cite{Carathodory1911} implies that $\mathcal{G}$ can be decomposed as a convex combination of $d^4 + 1$ extreme points $\mathcal{G}_{r} \in \textnormal{ext}(\textnormal{conv}(\mathcal{F}_{m}[\mathcal{E}])) \subseteq \mathcal{F}_{m}[\mathcal{E}]$. This means that there exists a probability distribution $\{p_r\}_{r=1}^{d^4 + 1}$ such that $\mathcal{G}$ can be decomposed as 
    \begin{align}
        \mathcal{G} &= \sum_{r = 1}^{d^4 + 1} p_{r} \, \mathcal{G}_{r} \quad \textnormal{ with } \quad \mathcal{G}_{r} = \sum_{i=1}^{m} \mathcal{K}_{ri} \circ \mathcal{I}_{ri}.
    \end{align}
    By Proposition \ref{prop:extreme_points_cf}, the subchannel decomposition $\{I_{ri}\}_{i=1}^{m} \in I_{m}[\mathcal{E}]$ of $\mathcal{E}$ corresponding to $\mathcal{G}_{r}$ can be chosen as extremal for all $r$. Now, Proposition \ref{prop:extremal_instrument} says that, by extremality, at most $d^4$ elements of $\{\mathcal{I}_{ri}\}_{i=1}^{m}$ can be non-zero so that, by discarding potential zero-maps $\mathcal{I}_{ri}$, we get 
    \begin{equation}
        \label{eq:d_bounded_class_future}
        \mathcal{G} = \sum_{r = 1}^{d^4 + 1} \sum_{i=1}^{d^4} \mathcal{K}_{ri} \circ (p_{r} \mathcal{I}_{ri}). 
    \end{equation}
    Readily, Eq.~\eqref{eq:d_bounded_class_future} confirms that $\mathcal{G} \in \mathcal{F}_{n}[\mathcal{E}]$ for some $n  \leq d^4(d^4 + 1)$ as $\{p_{r} I_{r}^{i}\}_{r,i}$ is a subchannel decomposition of $\mathcal{E}$.
 \end{proof}

\noindent\textit{Remark:}
One might wonder whether an extension of definition \eqref{eq:finite-classical-future} with $n \to \infty$, i.e., $\mathcal{F}_{\infty}[\mathcal{E}]$, can be outside $\mathcal{F}[\mathcal{E}]$. 
It can be shown, however, that every channel $\mathcal{G}$ in the classical future $\mathcal{F}_{\infty}[\mathcal{E}]$ with potential infinite classical memory arises as a limit point of a  converging sequence $\{\mathcal{G}_{n}\}_{n \in \mathbb{N}} \subseteq \mathcal{F}[\mathcal{E}]$. As $\mathcal{F}[\mathcal{E}]$ is closed, $\mathcal{G}$ also belongs to $\mathcal{F}[\mathcal{E}]$, and thus, $\mathcal{F}_{n}[\mathcal{E}]$ with some $n \leq d^4 (d^4 + 1)$.

\section{SDP relaxations of the classical future}
\label{app:memory_witness}

Am important insight for the characterization of the classical future $\mathcal{F}[\mathcal{E}]$ of a channel $\mathcal{E}$ can be derived from the Choi-Jamiołkowski representation. Recall that for each map $\mathcal{E}: \mathcal{L}(\mathbb{C}^d) \rightarrow \mathcal{L}(\mathbb{C}^d)$, one can associate the Choi operator $\mathcal{J}(\mathcal{E}) \in \mathcal{L}(\mathbb{C}^d \otimes \mathbb{C}^d)$ via
\begin{equation}
    \mathcal{J}(\mathcal{E}) = E^{\rm AB} := \sum_{i,j=0}^{d-1} \ketbra{i}{j} \otimes \mathcal{E}[\ketbra{i}{j}].
\end{equation}
We usually denote $\mathcal{J}(\mathcal{E})$ as $E^{\rm AB}$ to highlight the underlying bipartite structure of the operator. We adopt the convention that maps like channels and subchannels are denoted by caligraphic letters $\mathcal{E}$ and $\mathcal{I}$ while their associated Choi operators are written as capitalized latin letters $E^{\rm AB}$ and $I^{\rm AB}$. The components $E^{\rm AB}_{ij,kl}$ of the Choi operator in computational basis are connected to the map $\mathcal{E}$ via the equation
\begin{equation}
    \label{eq:choi_matrix_elements}
    E^{\rm AB}_{ij,kl} = \bra{k} \mathcal{E}[\ketbra{i}{j}] \ket{l}.
\end{equation}
A well known property of the Choi-Jamiołkowski isomorphism is that a map $\mathcal{E}: \mathcal{L}(\mathbb{C}^d) \rightarrow \mathcal{L}(\mathbb{C}^d)$ is a channel if and only if the associated Choi operator $E^{\rm AB}$ is positive semidefinite and $\tr_{\rm B}(E^{\rm AB}) = \mathds{1}^{\rm A}$.
For a proof of this property, see for instance \cite[Theorem 8]{Heinosaari2008}. 

Technically, the characterization of the classical future is based on the form of Choi operators of compositions of maps that is given by the link product as derived in Ref.~\cite{ziman08,Chiribella2009}.
    Let $\mathcal{C}: \mathcal{L}(\mathbb{C}^d) \rightarrow \mathcal{L}(\mathbb{C}^{d'})$ and $\mathcal{E}: \mathcal{L}(\mathbb{C}^{d'}) \rightarrow \mathcal{L}(\mathbb{C}^{d''})$ be maps with Choi operators $C^{\rm AD} = \mathcal{J}[\mathcal{C}]$ and $E^{\rm D'B} = \mathcal{J}[\mathcal{D}]$, respectively. Then, the Choi operator $\mathcal{J}[\mathcal{E} \circ \mathcal{C}]$ of the composition $\mathcal{E}\circ \mathcal{C}: \mathcal{L}(\mathbb{C}^d) \rightarrow \mathcal{L}(\mathbb{C}^{d''})$ is computed as 
    \begin{equation}
        \mathcal{J}[\mathcal{E} \circ \mathcal{C}]^{\rm AB} = {\ }^{\rm DD'}\bra{\mathbf{\Phi}^+} C^{\rm AD} \otimes E^{\rm D'B} \ket{\mathbf{\Phi}^+}^{\rm DD'}
        \label{eq:linked}
    \end{equation}
     where $\ket{\mathbf{\Phi}^+}^{\rm DD'} = \sum_{i=0}^{d'-1} \ket{i} \otimes \ket{i}$ denotes the unnormalized maximally entangled state.
This directly allows one to express the definition of the classical future in Choi's representation as follow.
\begin{Proposition}
\label{prop:class_future_sep}
    Let $\mathcal{E}, \mathcal{G} \in C_d$ be channels with associated Choi operators $E^{\rm AD}$ and $G^{\rm AB}$ such that $\mathcal{G} \in \mathcal{F}[\mathcal{E}]$. Then there exists an operator $X^{\rm ADD'B} \in \mathcal{L}((\mathbb{C}^{d})^{\otimes 4})$ of the form
    \begin{equation}
        X^{\rm ADD'B} = \sum_{i=1}^{n} I_{i}^{\rm AD} \otimes K_{i}^{\rm D'B}
    \end{equation}
    such that $G^{\rm AB} = {\ }^{\rm DD'}\bra{\mathbf{\Phi}^+} X^{\rm ADD'B} \ket{\mathbf{\Phi}^+}^{\rm DD'}$, where $G^{\rm AB}$ is the Choi operator of $\mathcal{G}$ and where $\ket{\mathbf{\Phi}^+} = \sum_{i=0}^{d-1} \ket{i} \otimes \ket{i}$. The matrices $I_{i}^{\rm AD}$ are bipartite positive semidefinite operators satisfying $\sum_{i=1}^{n} I_{i}^{\rm AD} = E^{\rm AD}$ and $K_{i}^{\rm D'B}$ are bipartite positive semidefinite operators such that $\tr_{\rm B}(K_{i}^{\rm D'B}) = \mathds{1}_{d}$ for all $i \in \{1,\dots,n\}$.
    Especially, the operator $X^{\rm ADD'B}$ is positive semidefinite, separable in the bipartition $(\rm{AD}|\rm{D'B})$ and fulfills \begin{equation}
        \tr_{\rm B}(X^{\rm ADD'B}) = E^{\rm AD} \otimes \mathds{1}_{d}.
    \end{equation}

\end{Proposition}
\begin{proof}
Let $\mathcal{G} \in \mathcal{F}[\mathcal{E}]$. This means that we can decompose $\mathcal{G}$ as $\mathcal{G} = \sum_{i=1}^{n} \mathcal{K}_{i} \circ \mathcal{I}_i$ with channels $\mathcal{K}_{i}$ and a subchannel decomposition $\{\mathcal{I}^i\}_{i=1}^{n}$ of $\mathcal{E}$. The Choi operator $G^{\rm AB} = \mathcal{J}[\mathcal{G}]$ can then, by linearity of the Choi isomorphism, be decomposed as
\begin{align}
    G^{\rm AB} &= \sum_{i=1}^{n} \mathcal{J}[\mathcal{K}_{i} \circ \mathcal{I}_i] = \sum_{i=1}^{n} {\ }^{\rm DD'}\bra{\mathbf{\Phi}^+} I_{i}^{\rm AD} \otimes K_{i}^{\rm D'B} \ket{\mathbf{\Phi}^+}^{\rm DD'}
\end{align}
where the second equality follows from Eq.~\eqref{eq:linked}. Furthermore, for the operator $X^{\rm ADD'B}$, given by 
\begin{equation*}
    X^{\rm ADD'B} = \sum_{i=1}^{n} I_{i}^{\rm AD} \otimes K_{i}^{\rm D'B},
\end{equation*}
 it holds that 
 \begin{align}
     \tr_{\rm B}(X^{\rm ADD'B}) &= \sum_{i=1}^{n} I_{i}^{\rm AD} \otimes \tr_{\rm B}(K_{i}^{\rm D'B}) 
     = \sum_{i=1}^{n} I_{i}^{\rm AD} \otimes \mathds{1}_{d} 
     = E^{\rm AD} \otimes \mathds{1}_{d}.
 \end{align}
\end{proof} 
The detection criterion presented in Proposition \ref{prop:class_future_sep} requires the characterization of separable quantum states that in general is known to be an NP-hard problem~\cite{Gurvits2003}. 
To arrive at a practically computable criterion,
we make use of the positive partial transposition (PPT) criterion \cite{Peres1996, Horodecki1996} which provides an outer approximation of the set of positive semidefinite separable operators. 
 \begin{Definition}
    \label{def:class_future_ppt_relaxation}
     Let $\mathcal{E}$ be a quantum channel. The $\rm{PPT}$-relaxation $\mathcal{F}_{\rm PPT}[\mathcal{E}]$ of the classical future $\mathcal{F}[\mathcal{E}]$ is defined by the set of channels $\mathcal{G}$ for which there exists a positive semidefinite operator $X^{\rm ADD'B}$ such that 
     \begin{align}
         {(X^{\rm ADD'B})}^{T_{\rm AD}} &\succcurlyeq 0 \\
         G^{\rm AB} &= {\ }^{\rm DD'}\bra{\mathbf{\Phi}^+} X^{\rm ADD'B} \ket{\mathbf{\Phi}^+}^{\rm DD'} \\
         \tr_{\rm B}(X^{\rm ADD'B}) &= E^{\rm AD} \otimes \mathds{1}_{d}
     \end{align}
     where $G^{\rm AB}$ and $E^{\rm AD}$ are the Choi operators of $\mathcal{G}$ and $\mathcal{E}$, respectively.
 \end{Definition}
The set $\mathcal{F}_{\rm PPT}[\mathcal{E}]$ is an outer approximation of the classical future in the sense that $\mathcal{F}[\mathcal{E}] \subseteq \mathcal{F}_{\rm PPT}[\mathcal{E}]$. Given a dynamics $\mathcal{E}(t)$ at two time steps, i.e., given $\mathcal{E}(t_1)$ and $\mathcal{E}(t_2)$, one can use the PPT relaxation of Definition \ref{def:class_future_ppt_relaxation} to formulate the following optimization problem 
\begin{align}
s_{\rm PPT}^* = \max_{s, \mathcal{G}} \; &s  \label{eq:ppt_quantum_memory_visibility}\\    
\textnormal{s.t. }&s \mathcal{E}(t_2) + (1-s) \mathcal{G} \in \mathcal{F}_{\rm PPT} [\mathcal{E}(t_1)] \nonumber \\
&\mathcal{G} \textnormal{ is CPTP}. \nonumber
\end{align}
The outcome of the optimization problem \eqref{eq:ppt_quantum_memory_visibility} gives rise to an upper bound to the problem \eqref{eq:convex_combi_future} from the main text.  

\section{Witnesses of quantum memory}
\label{app:quantum memory_witness}
As discussed in the main text, witnesses of quantum memory are observables whose value can detect the presence of quantum memory. The existence of such observables is guaranteed by the joint convexity of pairs of channels and corresponding channels in their classical future, that may be formalized as follows. 
\begin{theorem}
    \label{thm:graph_convexity}
    The set of pairs $\textnormal{Gr}(\mathcal{F}) := \{ (\mathcal{E}, \mathcal{G}) \in C_d \times C_d: \mathcal{G} \in \mathcal{F}[\mathcal{E}]\}$ of channels and associated channels in their classical future is convex. 
\end{theorem}
\begin{proof}
    Let $\mathcal{E}, \mathcal{G}, \mathcal{M}, \mathcal{N} \in C_d$ be channels such that $\mathcal{G} \in \mathcal{F}[\mathcal{E}]$ and $\mathcal{N} \in \mathcal{F}[\mathcal{M}]$. This means that 
    \begin{equation}
        \mathcal{G} = \sum_{i=1}^{n} \mathcal{K}_{i} \circ \mathcal{I}_{i} \textnormal{ with } \sum_{i=1}^{n} \mathcal{I}_{i} = \mathcal{E} \quad\quad \mathcal{N} = \sum_{i=1}^{n'} \mathcal{R}_{i} \circ \mathcal{S}_{i} \textnormal{ with } \sum_{i=1}^{n'} \mathcal{S}_{i} = \mathcal{M}
    \end{equation}
    Combinations then take the form
    \begin{equation}
        \lambda \mathcal{G} + (1-\lambda) \mathcal{N} = \sum_{i=1}^{n} \mathcal{K}_{i} \circ \lambda \mathcal{I}_{i} + \sum_{i=1}^{n'} \mathcal{R}_{i} \circ (1-\lambda)\mathcal{S}_{i}
    \end{equation}
    for all $\lambda \in \mathbb{R}$. If $\lambda \in [0,1]$, the collection of completely positive maps $\{\lambda \mathcal{I}_{i}\}_{i=1}^{n} \cup \{(1-\lambda)\mathcal{S}_{i}\}_{i=1}^{n'}$ is a subchannel decomposition of $\lambda \mathcal{E}+ (1-\lambda) \mathcal{M}$. This implies that $\lambda \mathcal{G} + (1-\lambda) \mathcal{N} \in \mathcal{F}[\lambda \mathcal{E}+ (1-\lambda) \mathcal{M}]$ so that $\textnormal{Gr}(\mathcal{F})$ is convex.   
\end{proof}

Theorem \ref{thm:graph_convexity} guarantees the existence of quantum memory witnesses defined as follows.
\begin{Definition} 
\label{def:q_mem_witness}
    A pair of hermitian observables $W_{1}^{\rm AD}, W_{2}^{\rm AB} \in \mathcal{L}(\mathbb{C}^{d} \otimes \mathbb{C}^{d})$ is called a witness of quantum memory if it holds that
    \begin{equation}
     \langle W \rangle_{\mathcal{E}_1, \mathcal{E}_2} :=   \tr(W_{1}^{\rm AD} E_{1}^{\rm AD}) + \tr(W_{2}^{\rm AB} E_{2}^{\rm AB}) \geq 0
    \end{equation}
    for all Choi operators $E_{1}^{\rm AD}$ and $E_{2}^{\rm AB}$ of channels $\mathcal{E}_{1}, \mathcal{E}_2 \in C_d$ such that $\mathcal{E}_2 \in \mathcal{F}[\mathcal{E}_{1}]$. 
\end{Definition}
To see how the observables can be measured in practice, notice that $W_1^{\rm AD}$ and $W_2^{\rm AB}$ can be expanded in some local basis via $W_{1(2)}^{\rm AD} = \sum_{i,j=0}^{d^{2}-1} w_{ij}^{1(2)} \, \sigma_{i} \otimes \sigma_{j}$ where $\{\sigma_{j}\}_{j=0}^{d^2 - 1}$ is some hermitian basis of the operator space $\mathcal{L}(\mathbb{C}^d)$, e.g., the basis of Pauli matrices for $d=2$. This can be used to rewrite the expectation value $\langle W \rangle_{\mathcal{E}}(t_2, t_1) := \langle W \rangle_{\mathcal{E}(t_1), \mathcal{E}(t_2)}$ as
\begin{equation}
    \langle W \rangle_{\mathcal{E}}(t_2, t_1) = \sum_{k=1,2} \sum_{i,j=0}^{d^2 - 1} w_{ij}^{k} \tr(\mathcal{E}(t_k)[\sigma_{i}^T] \sigma_{j})
\end{equation}
and allows to measure the witness on the dynamics $\mathcal{E}(t_k)$, $k=1,2$, at two times by preparing suitable states, evolving them by $\mathcal{E}(t_k)$ and finally performing a suitable measurement. A concrete example for such a quantum memory witness is provided below in Section \ref{app:dyn_2GA}. 
An important construction of witnesses can be derived from the PPT relaxation of the classical future in Definition \ref{def:class_future_ppt_relaxation}, as explained by the following Theorem. 
\begin{theorem}
    \label{thm:qm_witness_ppt}
    Let $W_{1}^{\rm AD}, W_{2}^{\rm AB} \in \mathcal{L}(\mathbb{C}^{d} \otimes \mathbb{C}^{d})$ be a pair of hermitian observables such that there exists positive semidefinite operators $Q^{\rm ADD'B}$ and $R^{\rm ADD'B}$ so that 
    \begin{equation}
        W_{1}^{\rm AD} \otimes \mathds{1}^{\rm D'B}/d +  W_{2}^{\rm AB} \otimes \Phi_{+}^{\rm DD'} = Q^{\rm ADD'B} + (R^{\rm ADD'B})^{T_{\rm D'B}}. 
    \end{equation}
    Then the pair $W_{1}^{\rm AD}, W_{2}^{\rm AB}$ is a witness of quantum memory. Here, $\Phi_{+}^{\rm DD'} := \ketbra{\mathbf{\Phi}^+}{\mathbf{\Phi}^+}^{\rm DD'} = 
    \sum_{i,j=0}^{d-1} \ketbra{i}{j} \otimes \ketbra{i}{j}$ is the unnormalized maximally entangled state.  
\end{theorem}
\begin{proof}
    Let $\mathcal{E}_{1}$ and $\mathcal{E}_{2}$ be a pair of channels such that $\mathcal{E}_2 \in \mathcal{F}[\mathcal{E}_{1}]$. By Proposition \ref{prop:class_future_sep} there exists a positive operator $X^{\rm ADD'B}$ such that $\tr_{\rm B}(X^{\rm ADD'B}) = E_{1}^{\rm AD} \otimes \mathds{1}^{\rm D'}$, $E_{2}^{\rm AB} = {\ }^{\rm DD'}\bra{\mathbf{\Phi}^+} X^{\rm ADD'B} \ket{\mathbf{\Phi}^+}^{\rm DD'}$ and $(X^{\rm ADD'B})^{T_{\rm D'B}} \geq 0$. Now, a straightforward calculation reveals 
    \begin{align}
        & \tr(W_{1}^{\rm AD} E_{1}^{\rm AD}) + \tr(W_{2}^{\rm AB} E_{2}^{\rm AB}) \\
        &= \tr[(W_{1}^{\rm AD} \otimes \mathds{1}_{\rm D'}/{d}) \, (E_{1}^{\rm AD} \otimes \mathds{1}_{\rm D'})] + \tr[W_{2}^{\rm AB} \, {\ }^{\rm DD'}\bra{\mathbf{\Phi}^+} X^{\rm ADD'B} \ket{\mathbf{\Phi}^+}^{\rm DD'}] \\
        &= \tr[(W_{1}^{\rm AD} \otimes \mathds{1}_{\rm D'}/d) \, \tr_{\rm B}(X^{\rm ADD'B})] + \tr[(W_{2}^{\rm AB} \otimes \Phi_{+}^{\rm DD'}) \, X^{\rm ADD'B}] \\
        &= \tr[(W_{1}^{\rm AD} \otimes \mathds{1}_{\rm D'B}/d +  W_{2}^{\rm AB} \otimes \Phi_{+}^{\rm DD'}) \, X^{\rm ADD'B}] \\
        &= \tr[(Q^{\rm ADD'B} + (R^{\rm ADD'B})^{T_{\rm D'B}}) \, X^{\rm ADD'B}] \\
        &= \underbrace{\tr(Q^{\rm ADD'B} X^{\rm ADD'B})}_{\geq 0} + \underbrace{\tr(R^{\rm ADD'B} (X^{\rm ADD'B})^{T_{\rm D'B}})}_{\geq 0} \geq 0 
    \end{align}
    showing that $W_{1}^{\rm AB}$ and $W_{2}^{\rm AD}$ form a witness of quantum memory. 
\end{proof}

\section{Spontaneous emission in two-level systems}
\label{app:dyn_two_level_sys}
A two-level quantum system interacting with the vacuum field undergoes spontaneous decay when it is initially in its excited state. In this process, the system transitions to it's ground state by emitting a photon with an energy of $\hbar\omega_e$, where $\omega_e$ represents the transition frequency between the excited and ground states, and $\hbar$ is Planck's constant. The total Hamiltonian of the two-level system and the environment reads
\begin{eqnarray}\label{TLS_gene_Hamiltonian}
    H &=& \hbar \omega_e \ketbra{e}{e} + \int_0^{\infty} d\omega \,  \hbar \omega a_{\omega}^{\dagger} a_{\omega} 
    +\hbar \int_0^{\infty} d\omega \sqrt{\frac{\Gamma(\omega)}{4\pi}} \left( a_{\omega}^{\dagger} \sigma_- + \sigma_+  a_{\omega}\right),
\end{eqnarray}
where $\Gamma(\omega)$ represents the spontaneous emission rate calculated using Fermi's golden rule. For generality, we allow this emission rate to depend on the frequency $\omega$,  as it is proportional to both the coupling strength between the two-level system and the field, and the spectral density of the environment. The lowering and raising operators of the two-level system are given by $\sigma_-= \ketbra{g}{e}$ and $\sigma_+=\ketbra{e}{g}$. The field operators $a_{\omega}$ correspond to modes labeled by the frequency $\omega$  and satisfy the commutation relation $[a_{\omega}, a_{\omega'}^{\dagger}]=\delta(\omega - \omega')$.

The aim is to derive the dynamics of the two-level system, starting from the global initial product state $\ket{\Psi(0)} = \left[e(0) \ket{e} + g(0)\ket{g}\right] \otimes \ket{\mathrm{vac}}$, where $\ket{\mathrm{vac}}$ represents the vacuum state of the mode continuum, and $g(0)$ and $e(0)$ are the initial amplitudes for the ground and excited states of the two-level system, respectively. Given that the field is initially in the vacuum state and interacts with the system in a weak and resonant manner, the single excitation assumption can be applied to the field \cite{Guo2017}. Consequently, the global state at any time $t$ is expressed using the pure state Ansatz
\begin{eqnarray}
\ket{\Psi(t)} = \int_0^\infty d\omega \, \alpha_{\omega}(t) a_{\omega}^{\dagger} \ket{g,\mathrm{vac}} + e(t) \ket{e,\mathrm{vac}} + g(0) \ket{g, \mathrm{vac}},
\end{eqnarray}
with $\alpha_{\omega} (0) = 0$. Notice that the amplitude of the $\ket{g,\mathrm{vac}}$ component is constant in time, because $H \ket{g, \mathrm{vac}} = 0$.
(The time evolution of an initially mixed atom state is straightforwardly obtained by mixing the solutions for two distinct pure initial states.)

From the Schrödinger equation $i\hbar \partial_t \ket{\Psi(t)} = H \ket{\Psi(t)}$, we obtain coupled equations for the probability amplitudes,
\begin{eqnarray}
    \dot{\alpha}_{\omega}(t) &=& -i\omega \alpha_{\omega}(t) - i e(t) \sqrt{\frac{\Gamma(\omega)}{4\pi}}, \label{EOM_env}\\
    \dot{e}(t) &=& - i\omega_e e(t) -i \int_0^{\infty} d\omega \sqrt{\frac{\Gamma(\omega)}{4 \pi}} \alpha_{\omega}(t). \label{EOM_sys}
\end{eqnarray}
We can formally solve the Eq.~\eqref{EOM_env} and substitute the solution into the Eq.~\eqref{EOM_sys}, which then becomes
\begin{eqnarray} \label{ex_amp_diff_eq}
    \dot{e}(t) = -i\omega_e e(t) -\int_0^t ds e(t-s) \chi(s),
\end{eqnarray}
with the bath correlation function at zero temperature
\begin{equation} \label{bcf_eq}
    \chi (t) = \frac{1}{4\pi}\int_0^{\infty} d\omega \Gamma(\omega) e^{-i\omega t},
\end{equation}
which depends on the details of the bath through $\Gamma(\omega)$. To solve the Eq.~\eqref{ex_amp_diff_eq}, we first perform the Fourier transform. The solution in frequency domain is
\begin{eqnarray}\label{ex_amp_freq_domain_eq}
    \tilde{e}(\omega) = \frac{ie(0)}{\omega - \omega_e + i\tilde{\chi}(\omega)},
\end{eqnarray}
where $\tilde{\chi}(\omega)$ has the definition $\tilde{\chi}(\omega)= \int_0^{\infty} dt \chi(t) e^{i \omega t}$. 
Taking the inverse Fourier transform, the solution of Eq.~\eqref{ex_amp_diff_eq} then has the form $e(t) = e(0)c(t)$,
with the decay factor 
\begin{equation}
    \label{eq:two_level_decay}
    c(t) =  \frac{i}{2\pi} \int_{-\infty}^{\infty} d\omega \frac{e^{-i\omega t}}{\omega -\omega_e +i \tilde{\chi}(\omega)}.
\end{equation}
Therefore, different types of baths lead to distinct spontaneous emission processes, each determined by the corresponding bath correlation function. 
In the corresponding spontaneous emission process, the state of the two-level system at time $t$ is represented by a reduced density matrix $\rho(t)$, obtained by performing a partial trace over the field degrees of freedom of the total system-field state, $\rho(t) = \tr_f\{|\Psi(t)\rangle \langle \Psi(t)|\}$. The density matrix elements $\rho_{xy}(t) = \bra{x}\rho(t)\ket{y}$ evolve according to
\begin{eqnarray}
    \rho_{ee} (t) &=&|c(t)|^2 \rho_{ee}(0), \label{ee_dyn} \\
    \rho_{eg} (t) &=&c(t) \rho_{eg}(0) = \rho_{ge}^*(t), \label{eg_dyn}\\
    \rho_{gg} (t) &=& \rho_{gg}(0) + (1-|c(t)|^2) \rho_{ee}(0). \label{gg_dyn}
\end{eqnarray}
Although these equations are based on the result for a pure product initial state, the decay function $c(t)$ as above in Eq.~\eqref{eq:two_level_decay} does not depend on the initial condition. Since any mixed initial state can be represented as a convex-linear combination of pure initial states, Eqs.~\eqref{ee_dyn}-\eqref{gg_dyn} hold true for any pure or mixed initial state.

Equivalently, we can express the time evolution as the action of a family of trace-preserving completely positive maps on the initial state, $ \mathcal{E}(t)[\rho(0)]$. These maps $\{ \mathcal{E}(t) | t \geq 0\}$ encapsulate how the system state changes over time due to its interaction with the field.

Based on \eqref{ee_dyn}-\eqref{gg_dyn}, we identify the action of the maps on individual operator basis elements as
 \begin{eqnarray}
      \mathcal{E}(t)[\ketbra{g}{g}] &=& \ketbra{g}{g}  \label{eq:two_level_channel_start},\\
     \mathcal{E}(t)[\ketbra{g}{e}] &=& c^{*}(t) \ketbra{g}{e}, \\
     \mathcal{E}(t)[\ketbra{e}{e}] &=& |c(t)|^2 \ketbra{e}{e} + \left(1- |c(t)|^2\right) \ketbra{g}{g}.\label{eq:two_level_channel_end}
 \end{eqnarray}
The Choi operator $E(t)$, see Section \ref{app:memory_witness}, associated with the channel $\mathcal{E}(t)$ is obtained by applying the map $\mathcal{E}(t)$ to one half of an unnormalized maximally entangled state $\ket{\mathbf{\Phi}^+} =  \ket{g} \otimes \ket{g} + \ket{e}\otimes \ket{e}$. The specific expression for the resulting Choi operator in computational basis is
\begin{eqnarray}
    \label{eq:choi_two_level}
    E (t) = \mathcal{E}(t) \otimes \openone  \left[\ketbra{\mathbf{\Phi}^+}{\mathbf{\Phi}^+}\right] = \begin{pmatrix}1 & 0 & 0 & c^*(t) \\
    0 & 0 & 0 & 0 \\
    0 & 0 & 1-|c(t)|^2 & 0 \\
    c(t) & 0 & 0 &|c(t)|^2 \end{pmatrix}.
    \label{choi_state}
\end{eqnarray}

In the following,  we give a complete characterization of the quantum memory properties of two-level spontaneous emission processes, and we present the scenario of a giant atom where spontaneous emission follows a non-exponential decay, which serves as a clear sign of non-Markovian dynamics. 

\subsection{Quantum memory characterization}
\label{app:two_level_qmc}
For a smooth connection to the mathematical framework discussed in Section \ref{app:memory_witness}, we replace the labels of the two levels by $\ket{g}$ and $\ket{e}$, as used in Section \ref{app:dyn_two_level_sys}, by the computational basis states $\ket{0}$ and $\ket{1}$, respectively.

With the Choi operator \eqref{eq:choi_two_level} of the general two-level system spontaneous emission process at hand, we can investigate the memory properties of the dynamical map. The distinction between Markovianity, classical memory and quantum memory for the two-level system can be characterized as follows.

\begin{theorem}
    The spontaneous emission dynamics of the two-level system with $\mathcal{E}(t)$ as given in Eq.~\eqref{eq:two_level_channel_start}-\eqref{eq:two_level_channel_end} can be implemented with classical memory if and only if 
    \begin{equation}
        \abs{c(t_1)} \geq \abs{c(t_2)}
    \end{equation}
    where $t_2 \geq t_1$ and with $c(t)$ as defined in Eq.~\eqref{eq:two_level_decay}. Furthermore, if the dynamics can be described with classical memory, it is even Markovian (CP-divisible) in the sense that there exists a channel $\mathcal{K}(t_2, t_1)$ such that 
    \begin{equation}
        \mathcal{E}(t_2) = \mathcal{K}(t_2, t_1) \circ \mathcal{E}(t_1). 
    \end{equation}
\end{theorem}
\begin{proof}
    For the first direction, let $\abs{c(t_2)} \leq \abs{c(t_1)}$. We now show that there exists a channel $\mathcal{K}(t_2, t_2)$ such that $\mathcal{E}(t_2) = \mathcal{K}(t_2, t_1) \circ \mathcal{E}(t_1)$. In fact, the map $\mathcal{K}(t_2, t_2) = \mathcal{C}(t_2) \circ \mathcal{C}(t_1)^{-1}$ must be associated to the Choi operator $K(t_2, t_1)^{\rm D'B}$, in computational basis explicitly given by 
    \begin{equation}
        K(t_2, t_1)^{\rm D'B} =  \begin{pmatrix} 
        1 & 0 & 0 & \frac{c(t_2)}{c(t_1)} \\
        0 & 0 & 0 & 0 \\
        0 & 0 & \frac{\abs{c(t_1)}^2 - \abs{c(t_2)}^2}{\abs{c(t_1)}^2} & 0 \\
        \frac{c(t_2)^*}{c(t_1)^*} & 0 & 0 & \frac{\abs{c(t_2)}^2}{\abs{c(t_1)}^2}
   \end{pmatrix}
    \end{equation}
    which is positive semidefinite if and only if $\abs{c(t_2)} \leq \abs{c(t_1)}$. It follows that the process is Markovian and hence may be implemented with classical memory.  

    For the other direction, assume that $\mathcal{E}(t_2) \in \mathcal{F}[\mathcal{E}(t_1)]$. 
    This means that there is a subchannel decomposition $\{\mathcal{I}_{i}(t_1)\}_{i=1}^{n}$ of $\mathcal{E}(t_1)$ and channels $\{\mathcal{K}_{i}(t_2, t_1)\}_{i=1}^{n}$ such that 
    \begin{equation}
    \label{eq:class_future-two_level_proof}
        \mathcal{E}(t_2) = \sum_{i=1}^{n} \mathcal{K}_{i}(t_2, t_1) \circ \mathcal{I}_{i}(t_1).
    \end{equation}
    Let us pick a single $i \in \{1,\dots,n\}$. Since $\mathcal{I}_{i}(t_1)$ is a subchannel of $\mathcal{E}(t_1)$, it's associated Choi operator satisfies $I_{i}(t_1)^{\rm AB} \leq E(t_1)^{\rm AB}$. Considering the operator $E(t_1)^{\rm AB}$, it can be seen that the subchannel $\mathcal{I}_{i}(t_1)$ must have $\ketbra{0}{0}$ as an eigenvector. Since $\mathcal{I}_{i}(t_2) := \mathcal{K}_{i}(t_2, t_1) \circ \mathcal{I}_{i}(t_1)$ is a subchannel of $\mathcal{E}(t_2)$, also $\mathcal{I}_{i}(t_2)$ has $\ketbra{0}{0}$ as an eigenvector. It follows that $\ketbra{0}{0}$ is also an eigenvector of $\mathcal{K}_{i}(t_2, t_1)$. Together with the assumption that $\mathcal{K}_{i}(t_2, t_1)$ is completely positive and trace-preserving,  the Choi operator $K_{i}(t_2, t_1)^{\rm D'B}$ of $\mathcal{K}_{i}(t_2, t_1)$ is necessarily of the form 
    \begin{equation}
    \label{eq:two_level_transition_choi}
        K_{i}(t_2, t_1)^{\rm AB} = \begin{pmatrix} 
        1 & 0 & 0 & k_{01,01} \\
        0 & 0 & 0 & 0 \\
        0 & 0 & k_{11,00} & k_{11,01} \\
        k_{01,01}^{*} & 0 & k_{11,01}^{*} & k_{11,11}
   \end{pmatrix}
    \end{equation}
    where $k_{lj,mr} = \bra{m} \mathcal{K}_{i} (t_2, t_2)[\ketbra{l}{j}]\ket{r}$. From the structure \eqref{eq:two_level_transition_choi} of the Choi operator $K_{i}(t_2, t_1)^{\rm D'B}$, one can deduce the equation
    \begin{equation}
        \bra{0} \mathcal{I}_{i}(t_2)[\ketbra{1}] \ket{0} = \bra{0} \mathcal{I}_{i}(t_1)[\ketbra{1}] \ket{0} + \underbrace{k_{11,00}  \bra{1} \mathcal{I}_{i}(t_1)[\ketbra{1}] \ket{1}}_{\geq 0}
    \end{equation}
    Hence, we have that 
    \begin{equation}
        \bra{0} \mathcal{I}_{i}(t_2)[\ketbra{1}{1}] \ket{0} \geq \bra{0} \mathcal{I}_{i}(t_1)[\ketbra{1}{1}] \ket{0}
    \end{equation}
    Since this inequality holds for all $i \in \{1,\dots,n\}$, it also holds for the sum of them so that 
    \begin{equation}
        1 - \abs{c(t_2)}^2 \geq 1 - \abs{c(t_1)}^2
    \end{equation}
    which is equivalent to $\abs{c(t_1)} \geq \abs{c(t_2)}$.
\end{proof}

\subsection{Example: Two-level giant atom dynamics}
\label{app:dyn_2GA}
\begin{figure}[t!] 
    \centering
        \includegraphics[width=0.45\linewidth]{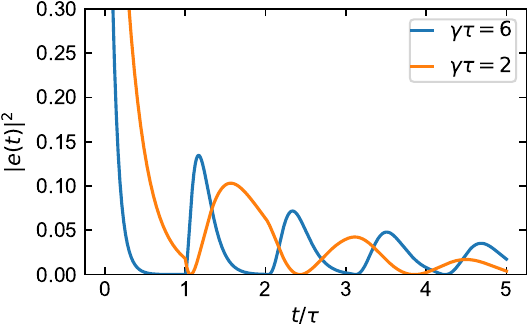}
    \caption{Time evolution of the excited-state probability of an initially excited two-level giant atom, for two different bath coupling rates $\gamma$ and fixed resonance frequency, $\omega_e \tau = 20 \pi$. Time is given in units of the delay time $\tau$.} 
    \label{fig:prob_dyn_2GA}
\end{figure}
We now analyze the dynamics of a two-level giant atom coupled to a one-dimensional surface acoustic wave (SAW) waveguide through two spatially separated coupling points, where the dynamics exhibits non-Markovian behavior.
As depicted in Fig.~\ref{fig:robustness_GA_sketch}(c) in the main text, we obtain the two-level system by eliminating the second level from the original three-level configuration. The atom interacts with both left- and right-propagating phonon fields at each coupling point. The distance between these two points, $\Delta x$, is significantly larger than the wavelength of the waveguide modes, resulting in a time delay $\tau = \Delta x/v_g$, which is incorporated as a phase factor into the system's dynamics. Here, $v_g$ is the velocity of the SAWs. At sufficiently low cryogenic temperatures, thermal SAW excitations in the vacuum state at frequencies $\omega \sim \omega_e$ are negligible. Assuming a flat coupling to the spectrum of SAW modes near the system resonance, as described in \cite{Gardiner2004, Guo2017}, the total system Hamiltonian is expressed as
\begin{eqnarray}
H &=& \hbar \omega_e \ketbra{e}{e} + \int_0^{\infty} d\omega \,  \hbar \omega \sum_{j=L,R} a_{j\omega}^{\dagger} a_{j\omega} 
+\hbar \sum_{j=L,R} \int_0^{\infty} d\omega \sqrt{\frac{\gamma}{4\pi}} \bigg[ a_{j\omega}^{\dagger} \sigma_- \left(e^{i k_{\omega} \Delta x/2} + e^{-i k_{\omega} \Delta x/2} \right) 
+ h.c. \bigg] \nonumber \\
&=&  \hbar \omega_e \ketbra{e}{e} + \int_0^{\infty} d\omega \,  \hbar \omega \sum_{j=L,R} a_{j\omega}^{\dagger} a_{j\omega} 
+\hbar \int_0^{\infty} d\omega \sqrt{\frac{\gamma}{\pi}} \cos \left( \frac{\omega \tau}{2}\right) \bigg[ \left(a_{L\omega} + a_{R\omega} \right)^{\dagger} \sigma_- + h.c. \bigg], 
\label{total_hamiltonian}
\end{eqnarray}
where the $j$ labels the left- and right-propagating fields and $\gamma$ represents the single-contact coupling rate. We can now switch to two standing-wave vacua via a beam splitter transformation, $a_{\pm \omega} = (a_{L\omega} \pm a_{R\omega})/\sqrt{2}$, and notice that the giant atom couples only to the $(+)$-modes. Hence, we can omit the uncoupled modes and rewrite the Hamiltonian in the form Eq.~\eqref{TLS_gene_Hamiltonian} with the effective spontaneous emission rate $\Gamma(\omega) = 8 \gamma \cos^2{\left(\omega \tau/2\right)}$. 
We then evaluate the bath correlation function \eqref{bcf_eq} and plug it into Eq.~\eqref{ex_amp_freq_domain_eq}, making the additional assumption that the Lamb shift is neglected. This leads to the excitation amplitude in Fourier space as
\begin{eqnarray}
    \tilde e(\omega) = \frac{ie(0)}{\omega-\omega_e +i\gamma/2 + i\gamma e^{i\omega \tau}/2} = \sum_{n=0}^\infty \frac{(-i\gamma/2)^n e^{in\omega \tau}}{(\omega-\omega_e + i\gamma/2)^{n+1}} = \sum_{n=0}^\infty \frac{(i\gamma/2)^n}{n!}e^{in\omega \tau} \frac{d^n}{d\omega^n} \left( \frac{1}{\omega-\omega_e + i\gamma/2} \right).
\end{eqnarray}
The geometric series expansion facilitates carrying out the inverse Fourier transform, which we express in terms of a decay factor relative to the initial amplitude,
\begin{eqnarray}
\label{eq:abs_c2_function}
    c(t) := \frac{e(t)}{e(0)} = \sum_{n=0}^{\infty} \Theta(t-n\tau) \frac{[-\gamma (t-n\tau)/2]^n}{n\,!}  e^{-i(\omega_e - i\gamma/2)(t-n\tau)}. \label{eq:c_fact_app}
\end{eqnarray}
The memory effect is encapsulated in the sum over past emission events occurring at times $t-n\tau$. These contributions are reabsorbed by the atom and interfere, leading to non-Markovian dynamics in which the atom's current state depends on its history.

Figure~\ref{fig:prob_dyn_2GA} illustrates the probability dynamics of a two-level giant atom in the excited state for varying coupling strengths with waveguide modes, highlighting the occurrence of revivals. These revivals indicate that the information initially lost from the system to the field returns over time. Each revival reaches its peak at approximately $ t \approx n\tau $ for $ n = 1, 2, 3, \dots$ and the offset depends on the parameters $\omega_e$, $\gamma$, and $\tau$. Notably, when $\gamma \tau \gg 1$, the system enters the over-damped regime, where the revivals quickly decay to zero and their peak values also decrease as $n$ increases. In this case, each revival is almost entirely due to the ''echo" of its direct predecessor, and not the other previous peaks.

As explained in the main text, we construct a witness for the detection of quantum memory of the two-level giant atom following the prescription of Theorem \ref{thm:qm_witness_ppt}. To this end, we must find two witness operators $W_{k=1,2} = \sum_{i,j=0}^{3} w_{ij}^{k} \, \sigma_{i} \otimes \sigma_{j}$ with parameters $w_{ij}^{k} \in \mathbb{R}$ such that the sum of the scalar product of $W_1$ with the Choi operator of $\mathcal{E}(t_1)$ and the scalar product of $W_2$ with the Choi operator of $\mathcal{E} (t_2)$ becomes negative. Written in terms of the channels, quantum memory is detected if
\begin{equation}
    \langle W \rangle_{\mathcal{E}(t)}(t_2, t_1) = \sum_{k=1,2} \sum_{i,j=0}^{3} w_{ij}^{k} \tr(\mathcal{E}(t_k)[\sigma_{i}^T] \sigma_{j}) < 0.
\end{equation}
where $\sigma_i$ are the usual Pauli matrices and $\sigma_0$ is the identity. 
For example, consider the system parameters $\omega_e \tau = 2 \pi$ and $\gamma \tau = 12$ as in Fig.~\ref{q_mem_robustness_GA}(a) and the two measurement time points $t_1 = 5.9/\gamma$ and $t_2 = t_1 + \Delta t$ with $\Delta t = 1.1/\gamma$. Then the witness parameters listed in Tables \ref{tab:witness_t1} and \ref{tab:witness_t2} detect quantum memory with $\langle W \rangle_{\mathcal{E}(t)}(t_2, t_1) < -0.0839$, providing a rather significant violation due to the comparison $\langle W \rangle_{\mathcal{E}(t)}(t_2, t_1) /(\tr(W_1) + \tr(W_2)) < -10^{-2}$.  

\begin{table}[t]
    \centering
    \begin{minipage}{0.4\linewidth}
    \caption{Witness coefficients $w_{ij}^{1}$ (rounded to three decimal places)}
    \begin{tabular}{c|cccc}
        \toprule
        \diagbox{Preparation $i$}{Measurement $j$} & $0$ & $1$ & $2$ & $3$ \\  
        \midrule
        $0$ & 1.012 & 0 & 0 & -0.988 \\  
        $1$ & 0 & 0 & 0 & 0 \\  
        $2$ & 0 & 0 & 0 & 0 \\  
        $3$ & 0.011 & 0 & 0 & 0.011 \\  
        \bottomrule
    \end{tabular}
    \label{tab:witness_t1}
    \end{minipage}
    \hfill
    \begin{minipage}{0.47\linewidth}
    \caption{Witness coefficients $w_{ij}^{2}$ (rounded to three decimal places)}
    \begin{tabular}{c|cccc}
        \toprule
        \diagbox{Preparation $i$}{Measurement $j$} & $0$ & $1$ & $2$ & $3$ \\  
        \midrule
        $0$ & 0.659 & 0 & 0 & -0.598 \\  
        $1$ & 0 & 0.106 & -0.184 & 0 \\  
        $2$ & 0 & -0.184 & -0.106 & 0 \\  
        $3$ & 0.367 & 0 & 0 & -0.375 \\  
        \bottomrule
    \end{tabular}
    \label{tab:witness_t2}
    \end{minipage}
\end{table}

\section{Spontaneous emission in $\Lambda$-type three-level systems}
\label{app:dyn_three_level_sys}
Consider a $\Lambda$-type three-level system consisting of three energy states: the ground state $\ket{g}$, the metastable state $\ket{s}$, and the excited state $\ket{e}$. Dipole-allowed transitions connect $\ket{e} \leftrightarrow \ket{g}$ and $\ket{e} \leftrightarrow \ket{s}$, which are coupled to a common radiation field, while direct transitions between the two lower states $\ket{g}$ and $\ket{s}$ are typically forbidden. In spontaneous emission, the system starts in the excited state and decays to either $\ket{g}$ or $\ket{s}$, emitting photons into the vacuum electromagnetic field. The total Hamiltonian reads
\begin{eqnarray}\label{eq:3TLS_gene_Hamiltonian}
    H = H_0 +
    \frac{\hbar}{\sqrt{4\pi}}\int_0^{\infty} d\omega \bigg[ a_{\omega}^\da \left(\sqrt{\Gamma_1(\omega)}|g\ra \la e| + \sqrt{\Gamma_2(\omega)} |s\ra \la e|\right) + a_{\omega} \left(\sqrt{\Gamma_1(\omega)}|e\ra \la g| + \sqrt{\Gamma_2(\omega)}|e\ra \la s|\right)\bigg],  
\end{eqnarray}
with the free Hamiltonian 
\begin{eqnarray}
    H_0 =  \hbar\omega_e \ketbra{e}{e} + \hbar\omega_s \ketbra{s}{s} + \int_0^{\infty} d\omega \, \hbar\omega  a_{\omega}^\da a_{\omega}.
\end{eqnarray}
Again, we consider the scenario where the bath can hold at most one photon excitation at a time. With this assumption, we make the following Ansatz for the general pure state of the system and field,
\begin{eqnarray}
    |\Psi(t)\rangle &=& e(t)|\mathrm{vac}, e\rangle + s(t)|\mathrm{vac}, s\rangle 
    + g(t)|\mathrm{vac}, g\rangle  \nonumber \\
    &+& \int_0^{\infty} d\omega \alpha_{\omega}(t)\hat{a}^{\dagger}_{\omega}|\mathrm{vac},g\rangle 
    +\int_0^{\infty} d\omega\beta_{\omega}(t) a^{\dagger}_{\omega}|\mathrm{vac},s\rangle.
\end{eqnarray}
Here, $|\mathrm{vac}\rangle$ denotes the vacuum state of the radiation field. 
The solution can be obtained in the same manner as in the previous Appendix~\ref{app:dyn_two_level_sys}, which results in $ g(t) = g(0)$, $s(t) = e^{-i\omega_s t} s(0)$, and $e(t) = d(t) e(0)$, with the decay factor
\begin{eqnarray}\label{decay_factor_3TLS}
    d(t)=\frac{i}{2\pi} \int_{-\infty}^{\infty} d\omega \frac{e^{-i\omega t}}{\omega -\omega_e +i \tilde{\chi}_1(\omega) +i \tilde{\chi}_2(\omega-\omega_s)} ,
\end{eqnarray}
where $\tilde{\chi}_{i=1,2}(\omega)$ are explicitly given by
\begin{eqnarray}\
     && \tilde{\chi}_1 (\omega) = \frac{1}{4\pi} \int_0^{\infty} \int_0^{\infty} dt d\omega' \Gamma_1(\omega') e^{i\left(\omega - \omega' \right) t} \label{bcf1_3TLS}\\
    &&\tilde{\chi}_2 (\omega-\omega_s) = \frac{1}{4\pi} \int_0^{\infty}\int_0^{\infty} dt d\omega' \Gamma_2(\omega') e^{i\left(\omega - \omega_s - \omega'\right) t}. \label{bcf2_3TLS}
\end{eqnarray}

For the $\Lambda$-type spontaneous emission process in a three-level system, we can construct a family of quantum channels $\mathcal{E}(t)$ that govern the time evolution of the atom's reduced state
\begin{eqnarray}
    \mathcal{E}(t)[\ketbra{g}{g}] &=& \ketbra{g}{g} \label{eq:three_level_channel_action_start} \\
    \mathcal{E}(t)[\ketbra{g}{s}] &=& e^{i\omega_s t} \ketbra{g}{s} \\
    \mathcal{E}(t)[\ketbra{g}{e}] &=& d^*(t) \ketbra{g}{e} \\
    \mathcal{E}(t)[\ketbra{s}{s}] &=& \ketbra{s}{s} \\
     \mathcal{E}(t)[\ketbra{s}{e}] &=& d^*(t) e^{-i\omega_s t} \ketbra{s}{e} \\
     \mathcal{E}(t)[\ketbra{e}{e}] &=& |d(t)|^2 \ketbra{e}{e} + \left(1- |d(t)|^2 -G(t)\right) \ketbra{s}{s} + G(t) \ketbra{g}{g}.\label{eq:three_level_channel_action_end} \\
\end{eqnarray}
Here, the function $G(t)$ is given by  
\begin{eqnarray}
    G(t)  =\frac{1}{4\pi} \int_0^{\infty} d\omega \Gamma_1(\omega) \left| \int_0^t dt'd(t')e^{-i\omega(t-t')} \right|^2.
\end{eqnarray}
The associated Choi operator 
can be written as
\begin{eqnarray}
    E (t) =
    \begin{pmatrix}
    1 & 0 & 0 & 0 & e^{i\omega_s t} & 0 & 0 & 0 & d^*(t) \\
    0 & 0 & 0 & 0 & 0 & 0 & 0 & 0 & 0 \\
    0 & 0 & 0 & 0 & 0 & 0 & 0 & 0 & 0 \\
    0 & 0 & 0 & 0 & 0 & 0 & 0 & 0 & 0 \\
    e^{-i\omega_s t} & 0 & 0 & 0 & 1 & 0 & 0 & 0 & d^*(t)e^{-i\omega_s t} \\
    0 & 0 & 0 & 0 & 0 & 0 & 0 & 0 & 0 \\
    0 & 0 & 0 & 0 & 0 & 0 & G(t) & 0 & 0 \\
    0 & 0 & 0 & 0 & 0 & 0 & 0 & 1-|d(t)|^2-G(t) & 0 \\
    d(t) & 0 & 0 & 0 & d(t)e^{i\omega_s t} & 0 & 0 & 0 & |d(t)|^2 \end{pmatrix}.
    \label{choi_state_3GA}
\end{eqnarray}

\subsection{Quantum memory characterization} 
\label{app:lambda_qm}

The presence of quantum memory for the $\Lambda$-type spontaneous emission process can be characterized as follows.
For a smooth connection to the mathematics discussed in Section \ref{app:memory_witness}, we will again switch the notation by changing the state labels of the three-level system to computational basis notation $\ket{0}, \ket{1}$ and $\ket{2}$, in energetically increasing order. 

\begin{theorem}
\label{thm:three_level_qm}
    The $\Lambda$-type three-level dynamics $\mathcal{E}(t):\mathbb{R} \rightarrow C_3$ as given in Eq.~\eqref{eq:three_level_channel_action_start}-\eqref{eq:three_level_channel_action_end} can be implemented with classical memory if and only if the inequalities
    \begin{align}
        G(t_2) &\geq G(t_1) \label{eq:3_qm_criterion_A} \\
        G(t_1) + \abs{d(t_1)}^2 &\geq G(t_2) + \abs{d(t_2)}^2. \label{eq:3_qm_criterion_B}
    \end{align}
    are fulfilled for all $t_{2} \geq t_1$. 
\end{theorem}

For the proof of Theorem \ref{thm:three_level_qm} we need the following technical Proposition.

\begin{Proposition}
\label{prop:three_level_subspace}
    Let $\mathcal{E}, \mathcal{G}: \mathcal{L}(\mathbb{C}^{3}) \rightarrow \mathcal{L}(\mathbb{C}^{3})$ be completely positive maps with corresponding Choi operators $E = \mathcal{J}(\mathcal{E})$ and $G = \mathcal{J}(\mathcal{G})$ that are both elements of the subspace of matrices $M$ that are in computational basis with standard order parametrized as 
    \begin{eqnarray}
    \label{eq:instrument_subspace}
    M = 
    \begin{pmatrix}
    \color{blue}{a} & 0 & 0 & 0 & \color{blue}{b} & 0 & \color{blue}{c} & \color{blue}{d} & \color{blue}{e} \\
    0 & 0 & 0 & 0 & 0 & 0 & 0 & 0 & 0 \\
    0 & 0 & 0 & 0 & 0 & 0 & 0 & 0 & 0 \\
    0 & 0 & 0 & 0 & 0 & 0 & 0 & 0 & 0 \\
    \color{blue}{b^*} & 0 & 0 & 0 & \color{blue}{f} & 0 & \color{blue}{g} & \color{blue}{h} & \color{blue}{j} \\
    0 & 0 & 0 & 0 & 0 & 0 & 0 & 0 & 0 \\
    \color{blue}{c^*} & 0 & 0 & 0 & \color{blue}{g^*} & 0 & \color{blue}{k} & \color{blue}{l} & \color{blue}{m} \\
    \color{blue}{d^*} & 0 & 0 & 0 & \color{blue}{h^*} & 0 & \color{blue}{l^*} & \color{blue}{n} & \color{blue}{p} \\
    \color{blue}{e^*} & 0 & 0 & 0 & \color{blue}{j^*} & 0 & \color{blue}{m^*} & \color{blue}{p^*} & \color{blue}{r} \end{pmatrix}.
\end{eqnarray}
We use colored symbols for better visibility of the sparsity structure of the matrix.  
    If there exists a channel $\mathcal{K} \in C_3$ such that $\mathcal{G} = \mathcal{K} \circ \mathcal{E}$ 
    then the inequalities $k_{G} \geq k_{E}$ and $n_{G} \geq n_{E}$ hold, where $k_{G}$ (or $k_E$) is the $k$-component of the Choi operator $G$ (or $E$) according to the parametrization \eqref{eq:instrument_subspace}.
\end{Proposition}
\begin{proof}
    We first show that the Choi operator $K^{\rm D'B}$ of the channel $\mathcal{K}$ is also an element of the subspace parametrized in Eq.~\eqref{eq:instrument_subspace}. Recall first that the matrix elements of the Choi operator may be related to the action of the corresponding map via Eq.~\eqref{eq:choi_matrix_elements}. From this, it can be observed that $\ketbra{0}$, $\ketbra{0}{1}$ and $\ketbra{1}{1}$ are eigenvectors of both $\mathcal{E}$ and $\mathcal{G}$ with eigenvalues $a_{E}, f_{E}, b_{E}$ and $a_{G}, f_{G}, b_{G}$, respectively. Hence, these operators are also eigenvectors of $\mathcal{K}$, and the corresponding eigenvalues are $a_{K} = 1$, $f_K = 1$ and $b_K$ since $\mathcal{K}$ is trace-preserving. Next, since one has 
    \begin{equation}
        \mathcal{E}[\ketbra{0}{2}] = c_{E} \ketbra{0}{0} + d_{E} \ketbra{0}{1} + e_{E} \ketbra{0}{2}
    \end{equation}
    together with a similar decomposition for $\mathcal{G}[\ketbra{0}{2}]$, and since $\ketbra{0}{0}$ and $\ketbra{0}{1}$ are eigenvectors of $\mathcal{E}, \mathcal{K}$ and $\mathcal{G}$, we can also decompose $\mathcal{K}[\ketbra{0}{2}]$ as 
    \begin{equation}
        \mathcal{K}[\ketbra{0}{2}] = c_{K} \ketbra{0}{0} + d_{K} \ketbra{0}{1} + e_{K} \ketbra{0}{2}. 
    \end{equation}
    The reasoning for the decomposition of $\mathcal{K}[\ketbra{0}{2}]$ can be similarly performed for $\mathcal{K}[\ketbra{1}{2}]$ from which it can be deduced that the Choi operator $K^{\rm D'B}$ also lives on the subspace parametrized by Eq.~\eqref{eq:instrument_subspace}.

    Since $\mathcal{K}$ is trace-preserving, we furthermore have that $c_{K} = h_{K} = 0$. 
    With this, we can compute $k_{G}$
    as
    \begin{align}
       k_{G} &= \bra{0}\mathcal{G}[\ketbra{2}{2}]\ket{0} = \bra{0}\mathcal{K} \circ \mathcal{E}[\ketbra{2}{2}]\ket{0} = k_{E}\underbrace{\bra{0}\mathcal{K}[\ketbra{0}{0}]\ket{0}}_{=a_{K}=1} + \underbrace{r_{E}}_{\geq 0} \underbrace{\bra{0}\mathcal{K}[\ketbra{2}{2}]\ket{0}}_{\geq 0} \geq k_{E}. 
    \end{align}
    In an analogous fashion, it can be shown that $n_{G} \geq n_{E}$. 
\end{proof}
Now, we are ready to show the characterization of quantum memory in the dynamics of the three-level giant atom.
\begin{proof}[Proof of Theorem \ref{thm:three_level_qm}]
    Assume that the inequalities \eqref{eq:3_qm_criterion_A} and \eqref{eq:3_qm_criterion_B} are satisfied for the two time points $t_1$ and $t_2$. Then it is straightforward to show, see the main text, that $\mathcal{E}(t)$ is Markovian (CP-divisible) as $\mathcal{E}(t_2) = \mathcal{K} \circ \mathcal{E}(t_1)$, with the map $\mathcal{K}$ being a quantum channel, if and only both the inequalities \eqref{eq:3_qm_criterion_A} and \eqref{eq:3_qm_criterion_B} are satisfied. As Markovian dynamics are classical, this implies the first direction of the Theorem.

    For the other direction, assume that the dynamics can be classically implemented, i.e., assume that $\mathcal{E}(t_2) \in \mathcal{F}[\mathcal{E}(t_1)]$ for all $t_2 \geq t_1$. By definition, this means that there is a subchannel decomposition $\{\mathcal{I}_{i}(t_2)\}_{i=1}^{n}$ of $\mathcal{E}(t_2)$ and $\{\mathcal{I}_{i}(t_1)\}_{i=1}^{n}$ of $\mathcal{E}(t_1)$ such that 
    \begin{equation}
        \mathcal{I}_{i}(t_2) = \mathcal{K}_{i}(t_2, t_1) \circ \mathcal{I}_{i}(t_1)
    \end{equation}
    for some channels $\mathcal{K}_{i}(t_2, t_2)$. By the description of the Choi operators $E(t)$ in Eq.~\eqref{choi_state_3GA}, the Choi operators $I_{i}(t_1)^{\rm AD}$ and $I_{i}(t_2)^{\rm AB}$ satisfy the assumption of the maps in Proposition~\ref{prop:three_level_subspace}. To see this, notice first that $I_{i}(t)^{\rm AD} \leq E(t)^{\rm AD}$ (in the sense of positive semidefinite matrices) due to the subchannel property. Next, let $X,Y \in \mathcal{L}(\mathbb{C}^d)$ be two positive semidefinite matrices such that $X \leq Y$. Then it is a standard fact that, if some diagonal element $Y_{aa}$ is zero, it holds that $X_{ab} = X_{ba}^{*} = 0$ for all $b \in \{1,\dots,d\}$. In computational basis, this restricts $I_{i}(t)^{\rm AD}$ to the subspace parametrized by Eq.~\eqref{eq:instrument_subspace}.
    
    Therefore, as an application of Proposition \ref{prop:three_level_subspace}, for all $i \in \{1,\dots,n\}$ it holds that
    \begin{align}
        \bra{0} \mathcal{I}_{i}(t_2)[\ketbra{2}{2}] \ket{0} &\geq \bra{0} \mathcal{I}_{i}(t_1)[\ketbra{2}{2}] \ket{0} \label{eq:lambda_type_proof_ineq_A}\\
        \bra{1} \mathcal{I}_{i}(t_2)[\ketbra{2}{2}] \ket{1} &\geq \bra{1} \mathcal{I}_{i}(t_1)[\ketbra{2}{2}] \ket{1}. \label{eq:lambda_type_proof_ineq_B}
    \end{align}
    Summing up the the inequalities \eqref{eq:lambda_type_proof_ineq_A} and \eqref{eq:lambda_type_proof_ineq_B} over the index $i$ leads to the inequalities \eqref{eq:3_qm_criterion_A} and \eqref{eq:3_qm_criterion_B}. 
\end{proof}

\subsection{Example: Three-level giant atom dynamics}
\label{app:dyn_3GA}
We examine a scenario where an artificial three-level giant atom \cite{Du2022,Sun2024} is coupled to a one-dimensional surface acoustic waveguide at multiple spatial points, as shown in Fig.~\ref{fig:robustness_GA_sketch}(c) in the main text. 
In this configuration, the transition between the excited and ground states is coupled to the waveguide with a single-contact coupling rate, denoted as $\gamma_1$. Simultaneously, the transition between the excited and metastable state occurs at a coupling rate $\gamma_2$.
When the atom, initially prepared in the excited state $|e\ra$, interacts with the vacuum acoustic field, it undergoes decay into one of the two lower states, emitting a phonon instead of a photon into the field, as described by the Hamiltonian under the rotating-wave approximation
\begin{equation}
H = H_0 + \hbar \sum_{j=R,L}\int_0^{\infty}d\omega \sqrt{\frac{\gamma_1}{\pi}} \cos{\left(\frac{\omega\tau}{2}\right)}\left(a_{j\omega}^\da |g\ra \la e| + h.c.\right) + \hbar \sum_{j=R,L}\int_0^{\infty}d\omega \sqrt{\frac{\gamma_2}{\pi}} \cos{\left(\frac{\omega\tau}{2}\right)}\left(a_{j\omega}^\da |s\ra \la e| + h.c.\right).
\end{equation}
The Hamiltonian can be further transformed into the general form \eqref{eq:3TLS_gene_Hamiltonian}, where the effective coupling rates for transitions $|e\rangle \leftrightarrow |g\rangle$ and $|e\rangle \leftrightarrow |s\rangle$ are given by $\Gamma_{i=1,2}(\omega) = 8 \gamma_i \cos^2{\left(\omega \tau/2\right)}$, similar to the case of the two-level giant atom. Based on this, the bath correlation functions \eqref{bcf1_3TLS}-\eqref{bcf2_3TLS} can be computed. Substituting these into Eq.~\eqref{decay_factor_3TLS}, we make an additional approximation by neglecting the divergent Lamb shift term, leading to
\begin{eqnarray}
    d(t)=\sum_{n=0}^{\infty} \Theta(t-n\tau) \frac{[-\bar{\gamma}_2(t-n\tau)]^n}{n\,!}e^{-i(\omega_e-i\bar{\gamma}_1)(t-n\tau)} \nonumber,
\end{eqnarray}
where $\bar{\gamma}_1= \left(\gamma_1+\gamma_2\right)/2$, and $\bar{\gamma}_2=\left(\gamma_1 e^{-i\omega_s \tau}+\gamma_2\right)/2$. 

As for the two-level giant atom, it is possible to detect quantum memory for three levels under reasonable experimental assumptions, in particular with settings that do not allow for a tomographically complete reconstruction of the channels. We assume that preparations and measurements are possible that only project onto the set of $11$ different states $\{\ket{\psi_j}\}_{j=1}^{11} = \{\ket{g}, \ket{s}, \ket{e}, \frac{1}{\sqrt{2}}(\ket{g} + \ket{e}), \frac{1}{\sqrt{2}}(\ket{s} + \ket{e}), \frac{1}{\sqrt{2}}(\ket{g} - \ket{e}), \frac{1}{\sqrt{2}}(\ket{s} - \ket{e}), \frac{1}{\sqrt{2}}(\ket{g} + i \ket{e}), \frac{1}{\sqrt{2}}(\ket{g} - i \ket{e}), \frac{1}{\sqrt{2}}(\ket{s} + i \ket{e}), \frac{1}{\sqrt{2}}(\ket{s} - i \ket{e})\}$.
This means, we construct a witness $W$ that is evaluated as 
\begin{equation}
    \label{eq:qutrit_witness_form}
    \langle W \rangle_{\mathcal{E}(t)}(t_2, t_1) = \sum_{\alpha = 1}^{2} \sum_{i,j=1}^{11} w_{ij}^{\alpha} \bra{\psi_j} \mathcal{E}(t_{\alpha})[(\ketbra{\psi_i}{\psi_i})^T] \ket{\psi_j}.
\end{equation}
with real parameters $w_{ij}^{\alpha}$ that have to be determined.

Using a witness of the form \eqref{eq:qutrit_witness_form} constructed from the PPT relaxation for quantum memory witnesses in Theorem~\ref{thm:qm_witness_ppt}, we obtained a witness violation $\langle W \rangle_{\mathcal{E}(t)}(t_2, t_1) < -0.016$, see Tables \ref{tab:witness_t1} and \ref{tab:witness_t2}, at the optimally robust time points given $(\gamma_1 + \gamma_2)t_1 \approx 6$ and $(\gamma_1 + \gamma_2)t_2 \approx 6.92$ using the same parameters as for Fig.~\ref{q_mem_robustness_GA} of the main text. While it is interesting that quantum memory detection is possible under the assumption of incomplete information, the witness is normalized to $\sum_{\alpha = 1} \sum_{i,j=1} w_{ij}^{\alpha} \approx 27.039$ and hence it should be noted that the violation is relatively weak in regard of the noise present in current experimental setups.

\begin{table}[ht!]
    \centering
    \scalebox{0.9}{
    \begin{tabular}{c|ccccccccccc}
         \diagbox{Preparation $i$}{Measurement $j$}& 1 & 2 & 3 & 4 & 5 & 6 & 7 & 8 & 9 & 10 & 11 \\
        \midrule
        1  & 0.009 & 0.842 & 0.284 & 0.146 & 0.492 & 0.146 & 0.492 & 0.146 & 0.146 & 0.563 & 0.563 \\
        2  & 0.742 & 0.059 & 0.240 & 0.491 & 0.090 & 0.491 & 0.090 & 0.491 & 0.491 & 0.150 & 0.150 \\
        3  & -0.518 & -0.561 & 0.437 & -0.041 & -0.171 & -0.041 & -0.172 & -0.041 & -0.041 & -0.062 & -0.062 \\
        4  & -0.255 & 0.141 & 0.360 & 0.050 & 0.160 & 0.056 & 0.160 & 0.053 & 0.053 & 0.250 & 0.250 \\
        5  & 0.241 & -0.111 & 0.230 & 0.235 & 0.000 & 0.235 & 0.004 & 0.235 & 0.235 & 0.060 & 0.060 \\
        6  & -0.255 & 0.141 & 0.360 & 0.056 & 0.160 & 0.050 & 0.160 & 0.053 & 0.053 & 0.250 & 0.250 \\
        7  & 0.242 & -0.110 & 0.229 & 0.235 & 0.004 & 0.235 & -0.000 & 0.235 & 0.235 & 0.059 & 0.059 \\
        8  & -0.255 & 0.141 & 0.360 & 0.053 & 0.160 & 0.053 & 0.160 & 0.053 & 0.053 & 0.250 & 0.250 \\
        9  & -0.255 & 0.141 & 0.360 & 0.053 & 0.160 & 0.053 & 0.160 & 0.053 & 0.053 & 0.250 & 0.250 \\
        10 & 0.112 & -0.251 & 0.339 & 0.225 & -0.041 & 0.225 & -0.041 & 0.225 & 0.225 & 0.042 & 0.046 \\
        11 & 0.112 & -0.251 & 0.339 & 0.225 & -0.041 & 0.225 & -0.041 & 0.225 & 0.225 & 0.046 & 0.042 \\
    \end{tabular}
    }
    \caption{Witness coefficients $w_{ij}^1$ for preparations and measurements at the first time point (rounded to three decimal places).}
    \label{tab:witness1}
    \end{table}

    \begin{table}[ht!]
    \centering
    \scalebox{0.9}{  
        \begin{tabular}{c|ccccccccccc}
           \diagbox{Preparation $i$}{Measurement $j$} & 1 & 2 & 3 & 4 & 5 & 6 & 7 & 8 & 9 & 10 & 11 \\
            \midrule
            1  & -0.272 & -0.093 & 0.681 & 0.205 & 0.150 & 0.205 & 0.097 & 0.205 & 0.205 & 0.256 & 0.332 \\
            2  & -0.199 & -0.158 & 0.677 & 0.239 & 0.110 & 0.239 & 0.071 & 0.239 & 0.239 & 0.231 & 0.288 \\
            3  & -0.130 & -0.139 & 0.129 & -0.000 & -0.090 & -0.000 & 0.016 & -0.000 & -0.000 & 0.072 & -0.082 \\
            4  & -0.201 & -0.116 & 0.405 & -0.014 & 0.030 & 0.218 & 0.056 & 0.507 & -0.303 & 0.164 & 0.125 \\
            5  & -0.106 & -0.094 & 0.318 & 0.106 & 0.225 & 0.106 & -0.159 & 0.106 & 0.106 & -0.166 & 0.390 \\
            6  & -0.201 & -0.116 & 0.405 & 0.218 & 0.030 & -0.014 & 0.056 & -0.303 & 0.507 & 0.164 & 0.125 \\
            7  & -0.158 & -0.133 & 0.424 & 0.133 & -0.159 & 0.133 & 0.238 & 0.133 & 0.133 & 0.433 & -0.143 \\
            8  & -0.201 & -0.116 & 0.405 & -0.303 & 0.030 & 0.507 & 0.056 & 0.102 & 0.102 & 0.164 & 0.125 \\
            9  & -0.201 & -0.116 & 0.405 & 0.507 & 0.030 & -0.303 & 0.056 & 0.102 & 0.102 & 0.164 & 0.125 \\
            10 & -0.126 & -0.120 & 0.326 & 0.100 & 0.288 & 0.100 & -0.245 & 0.100 & 0.100 & 0.323 & -0.116 \\
            11 & -0.203 & -0.177 & 0.480 & 0.138 & -0.268 & 0.138 & 0.331 & 0.138 & 0.138 & -0.020 & 0.323 \\
        \end{tabular}
    }
    \caption{Witness coefficients $w_{ij}^2$ for preparations and measurements at the second time point (rounded to three decimal places).}
    \label{tab:witness2}
    \end{table}

\end{document}